\documentclass[a4paper,USenglish]{lipics-v2021}

\pdfoutput=1 
\hideLIPIcs  

\title{Black Hole Search in Dynamic Tori}

\author{Adri {Bhattacharya}}{Indian Institute of Technology Guwahati, Assam, India}{a.bhattacharya@iitg.ac.in}{[OrcidId:0000-0003-1517-8779]}{}

\author{Giuseppe F. Italiano}{Luiss University, Rome, Italy}{gitaliano@luiss.it}{[OrcidId: 0000-0002-9492-9894]}{}

\author{Partha Sarathi Mandal}{Indian Institute of Technology Guwahati, Assam, India}{psm@iitg.ac.in}{0000-0002-8632-5767}{}

\authorrunning{Bhattacharya et al.}
 

\funding{\textit{Adri Bhattacharya:} Supported by CSIR, Govt. of India, Grant Number: 09/731(0178)/2020-EMR-I}

\usepackage{amsfonts, amssymb, amsmath,pifont,float,caption,multicol}
\usepackage{algorithm,algpseudocode}
\usepackage[linesnumbered,ruled,vlined,algo2e]{algorithm2e}
\newtheorem*{blub}{Note}
\newtheorem{blab}{Corollary}
\newtheorem{defn}{Definition}
\newtheorem{obs}{Observation}
\newtheorem{lemm}{Lemma}
\nolinenumbers
\begin{document}

\maketitle

\begin{abstract}
We investigate the black hole search problem by a set of mobile agents in a dynamic torus. Black hole is defined to be a dangerous stationary node which has the capability to destroy any number of incoming agents without leaving any trace of its existence. A torus of size $n\times m$ ($3\leq n \leq m$) is a collection of $n$ row rings and $m$ column rings, and the dynamicity is such that each ring is considered to be 1-interval connected, i.e., in other words at most one edge can be missing from each ring at any round. The parameters which define the efficiency of any black hole search algorithm are: the number of agents and the number of rounds (or \textit{time}) for termination. We consider two initial configurations of mobile agents: first, the agents are co-located and second, the agents are scattered. In each case, we establish lower and upper bounds on the number of agents and on the amount of time required to solve the black hole search problem.   
\end{abstract}
\section{Introduction}

Given a network and a set of mobile agents, the black hole search problem (also termed as BHS problem) consists of
locating a malicious stationary node which has the power to eliminate any number of incoming agents without leaving any trace of its existence. This problem is not new and it readily has many real life implications. For example, the black hole may be a node 
infected with a virus in a computer network, and in order to make the network safe the infected node 
should be located for further actions. The first task for any set of mobile agents ought to be to locate the black hole. To accomplish this task, 
at least one agent needs to visit the node; 
we aim at an efficient BHS algorithm, where 
the minimum 
number of agents gets consumed by the black hole and so that at least one agent must remain alive in order to locate the black hole within finite time. This problem has been extensively studied in networks which are static, see, e.g., 
\cite{TimeOptimalBHSRingBalamohan, BHSscatteredSynchronousRingchalopin,  BHSdirectedGraphczyzowicz, BHSsynchronousTreeczyzowicz,BHSCommonInterconnectionNetworkDobrev,BHSStaticArbNetworkDobrev,BHSAsyncRingTokensDobrev,BHSscatteredUnorientedRingTokenDobrev,BHSPebblesFlocchini}. Recently, research on black hole search problem has been mainly focused on dynamic networks; in particular, the most relevant dynamic networks studied are 
\textit{time-varying graphs}. These networks work on 
temporal domains, which are mainly considered to be discrete time steps. More precisely, 
the network is 
a collection of static graphs, in which some edges may disappear or reappear at each discrete time step, while 
the vertex set is fixed,
with the additional constraint that at each time step the underlying graph remains 
connected (also termed as \textit{1-interval connected}). Presently, apart from the black hole search on a dynamic ring \cite{BHSDynRingLuna} and on a dynamic cactus \cite{BHSdynamicCactusAdri}, nothing much is known about the black hole search problem 
on dynamic 
networks. 

In this paper, we 
investigate 
this problem on a dynamic torus of size $n \times m$ (where each ring is 1-interval connected and without loss of generality $3\leq n \leq m$), where a set of agents synchronously perform the same execution, with the goal of locating the black hole. We study two types of initial configurations of agents. 
In the first configuration, all agents are 
initially located at the same node; in the second configuration, the agents are scattered along the nodes of the underlying network. In both configurations, all the nodes where agents are 
initially located are not \textit{dangerous}, i.e., they do not contain the black hole (they are \textit{safe}). Our primary objective is to design an efficient BHS algorithm such that: (a) within a finite time at least one agent survives, and (b) it gains 
knowledge of the black hole location.

\section{Related Works and Our Contribution}

\subsection{Related work}

Network exploration 
by mobile agents is one of the fundamental problems in this domain, and it was
first introduced by Shannon \cite{ExplorationFirstWorkShannon}. After his pioneering 
work, this problem has been extensively studied in various topologies such as rings \cite{RingExplorationNagahama}, trees \cite{TreeExplorationDas}, general graphs \cite{GraphExplorationCohen} under different models of communication (particularly, pebbles \cite{ExplorationPebbleDisser} and whiteboard \cite{WhiteboardExplorationSudo}), synchrony (synchronous \cite{GraphExplorationCohen}, semi-synchronous \cite{SemiSynchronousExplorationBrandt} and asynchronous \cite{AsyncExplorationRingFlocchini}) and both in static \cite{GraphExplorationCohen} as well as dynamic networks (tori \cite{DynTorusExpGotoh} and general graphs \cite{DynamicGeneralGraphExplorationGotoh}). 

The black hole search (BHS) problem is a special version of the exploration problem, where in the worst case the underlying network needs to be explored in order to locate the black hole position. This problem was
first introduced by Dobrev et al. \cite{BHSStaticArbNetworkDobrev}, and after that has received a lot of attention: indeed, it
has been studied for directed \cite{BHSdirectedGraphczyzowicz} as well as undirected graphs \cite{BHSscatteredSynchronousRingchalopin}, and for different underlying networks, such as rings \cite{BHSscatteredSynchronousRingchalopin}, tori \cite{BHSscatteredStaticTorusChalopin}, trees \cite{BHSsynchronousTreeczyzowicz} and general graphs \cite{BHSStaticArbNetworkDobrev}. In addition, different communication models have 
been considered
for this problem, including
`Enhanced Token' \cite{BHSTokenDobrev}, `Pure Token' \cite{BHSPebblesFlocchini} and whiteboard \cite{BHSWhiteboardDobrev}. Moreover, this problem has also been explored for different initial agent configurations. In particular, Shi et al. \cite{BHSTokenShi} showed
that, when the agents are co-located, a minimum of 2 co-located agents communicating via tokens can solve the BHS problem in hypercube, torus and complete network with $\Theta(n)$ moves, whereas in the case where
$k$ agents ($k>3$) are scattered, then with only 1 token per agent it is shown that BHS can be solved in $O(k^2n^2)$ moves. All these above papers discuss black hole search in a static network, and very little is known about the problem in dynamic networks. Di Luna et al. \cite{BHSDynRingLuna} first investigated 
this problem in a dynamic ring, and they showed 
that when the agents are co-located, then in face-to-face communication with 3 agents there is an optimal algorithm that works in $\Theta(n^2)$ moves and $\Theta(n^2)$ rounds (where $n$ is the size of the ring). Next, with whiteboard communication, they 
reduced the complexity to $\Theta(n^{1.5})$ rounds and $\Theta(n^{1.5})$ moves. Lastly, when the agents are initially scattered and each node has a whiteboard, then again with 3 agents they showed that
at least $\Theta(n^2)$ moves and $\Theta(n^2)$ rounds are required for any BHS algorithm. In each case, they gave
an optimal algorithm. Next, Bhattacharya et al. \cite{BHSdynamicCactusAdri} studied the BHS problem in a dynamic cactus graph, 
and proposed an agent optimal algorithm when at most one edge can be dynamic;
in the case when at most $k$ $(>1)$ edges can be dynamic, they
proposed a lower bound of $k+2$ and an upper bound of $2k+3$ agents. 

In this paper, we further investigate the 
BHS problem in a dynamic torus, with the aim of providing
an efficient BHS algorithm. To the best of our knowledge, this is the first work where the BHS problem is explored in the case of a dynamic torus. Previously, Gotoh et al. \cite{DynTorusExpGotoh} studied the exploration problem under link presence detection and no link presence detection in dynamic tori, whereas Chalopin et al. \cite{BHSscatteredStaticTorusChalopin} studied the BHS problem in a static torus and gave tight bounds on the number of agents and tokens when the agents are initially scattered.

\subsection{Our Contribution}
We investigate the BHS problem in a dynamic torus for two initial configurations: first, when the set of agents are initially co-located, and next, when the agents can be initially scattered in different nodes. When the agents are initially co-located, we provide the following results.
\begin{itemize}
    \item We establish the impossibility of correctly locating the black hole with $n+1$ agents.
    \item We show that with $n+c$ (where $c \geq 2$ and $c\in \mathbb{Z}^+$) co-located agents, any BHS algorithm requires at least $\Omega(m\log n)$ rounds.
    \item With $n+3$ agents we present 
    a BHS algorithm that works in $O(nm^{1.5})$ rounds.
    \item Next, with $n+4$ agents we present 
    an improved BHS algorithm that works in $O(mn)$ rounds.
    \end{itemize}
The following results are obtained when the agents are initially scattered.
    \begin{itemize}
        \item We establish the impossibility of correctly locating the black hole with $n+2$ agents.
        \item We show 
        that with $k=n+c$ (where $c \geq 3$ and $c\in \mathbb{Z}^+$) scattered agents, any BHS algorithm requires $\Omega(mn)$ rounds.
        \item With $n+6$ agents we present 
        a BHS algorithm that works in $O(nm^{1.5})$ rounds.
        \item Lastly, with $n+7$ agents we present 
        a round optimal BHS algorithm that works in $O(mn)$ rounds.
    \end{itemize}

\vspace{-0.5cm}

 \begin{table}[H]
 \centering
 \begin{tabular}[t]{|c|c|c|c|c|}
 \hline

  IC & Bound &\# Agents &  Rounds & Results\\
 \hline

 Colocated& LB & $n+2$ & $\Omega(m\log n)$& Cor \ref{corollary:colocatedAgentLB} \& Thm \ref{theorem:colocatedLBComplexity}\\\cline{2-5}

 &UB&$n+3$&$O(nm^{1.5})$ & Thm \ref{theorem:ComplexityAlg4} \\\cline{2-5}

 &UB& $n+4$& $O(nm)$& Thm \ref{theorem:n+4colocatedRoundComplexity} \\

 \hline

 Scattered &LB& $n+3$& $\Omega(nm)$& Cor \ref{corollary:scatteredAgentLB} \& Thm \ref{theorem:scatteredRoundLB}\\\cline{2-5}

 & UB& $n+6$ &  $O(nm^{1.5})$ & Thm \ref{theorem:n+6scatteredRoundComplexity}\\\cline{2-5}

 &UB &$n+7$&$O(nm)$ & Thm \ref{theorem:n+7scatteredRoundComplexity}\\
 \hline
 \end{tabular}
 \caption{Summary of Results where LB, UB and IC represent lower bound, upper bound and initial configuration of the agents, respectively.}
 \end{table}

\vspace{-0.5cm}
\noindent\textbf{Organisation:} The remainder of the paper is organised as follows. In Sections \ref{section3} and \ref{section-LB}, we explain the model and prove the lower bound results. Next, in Section \ref{section-preliminaries}, we discuss 
some preliminary notation and basic subroutines which will be used by
our algorithms. Further, in Sections \ref{section-colocated} and \ref{section-scattered}, we present and analyse our algorithms for the co-located and scattered case. 
Finally, we list some concluding remarks in Section \ref{section-conclusion}.

\section{Model and Problem Definition }\label{section3}
\subsection{Graph Model}

The dynamic graph is modelled as a time-varying graph (or formally known as \textit{temporal graph}) $\mathcal{G}=(G,V,E,\mathbb{T},\rho)$, where $V$ is the set of vertices (or nodes), $E$ is the set of edges in $G$, $\mathbb{T}$ is defined to be the \textit{temporal domain}, which is defined to be $\mathbb{Z}^+$ as in this model we consider discrete time steps, also $\rho:E\times \mathbb{T} \rightarrow \{0,1\}$ is defined as the \textit{presence} function, which indicates the presence of an edge at a given time. The graph $G=(V,E)$ is the underlying static graph of the dynamic graph $\mathcal{G}$, also termed as \textit{footprint} of $\mathcal{G}$. More specifically, the footprint $G=(V,E)$ is a torus of size $n\times m$, where $n$ represents the number of rows and $m$ represents the number of columns, we define $V=\{v_{i,j} ~|~ 0 \le i \le n-1, 0 \le j \le m-1\}$ and $E$ is the set of edges, where the horizontal and vertical edges are $\{(v_{i,j}, v_{i,j+1\mod{m}})\}$ and $\{(v_{i,j}, v_{i+1\mod{n},j})\}$, respectively (refer Fig. \ref{fig:example_torus}). By the node $v_{i,j}$ we invariably mean $v_{i\mod{n}, j\mod{m}}$ and these modulus functions are ignored further in this paper. In order to restrict self loop or multiple edges, without loss of generality we assume $3\le n \le m$. A row ring $R_i$ (resp, a column ring $C_j$) is the subgraph of $G$ induced by the set of vertices $\{v_{i,j} ~|~ 0\le j \le m-1\}$ (resp, $\{v_{i,j}~|~0\le i \le n-1\}$). In this paper, we consider our temporal graph $\mathcal{G}$ to be a oriented dynamic torus. The adversary has the ability to make an edge reappear or disappear at any particular time step with the added constraint that, irrespective of how many edges disappear or reappear, each row and column ring at any time step must be connected; in other words each row and column ring in $\mathcal{G}$ is \textit{1-interval connected} (so at any time, the adversary can make at most one edge disappear from each row and column ring, in order to maintain the 1-interval connectivity property). A disappeared edge is termed as a \textit{missing edge} in this paper.

Every node $v_{i,j}\in G$ is labelled by a unique Id $(i,j)$, whereas each node in $G$ has 4 \textit{ports} adjacent to it, where the ports corresponding to the edges $(v_{i,j},v_{i,j-1})$, $(v_{i,j},v_{i,j+1})$, $(v_{i,j},v_{i-1,j})$, $(v_{i,j},v_{i+1,j})$, are denoted by \textit{west, east, south, north}, respectively. In addition, corresponding to each port of a node $v_{i,j}$ of $G$ a \textit{whiteboard} of storage of $O(1)$-bits is placed. The purpose of the whiteboard is to store and maintain certain information such as the node Id or agent Id or the agent's course of traversal (depending on the amount of storage the whiteboard can store). Any incoming agent can read the existing information or write any new information corresponding to a port along which it travels to the next node. Fair mutual exclusion to all incoming agents restricts concurrent access to the whiteboard. The network $G$ has a malicious node or unsafe node (refer BH in Fig. \ref{fig:example_torus}), also termed as a \textit{black hole}, which vanishes any incoming agent without leaving any of its trace. The remaining nodes in $G$ are not malicious, hence they are termed as \textit{safe nodes}.

\begin{figure}
  	\centering
  	\includegraphics[scale=0.8]{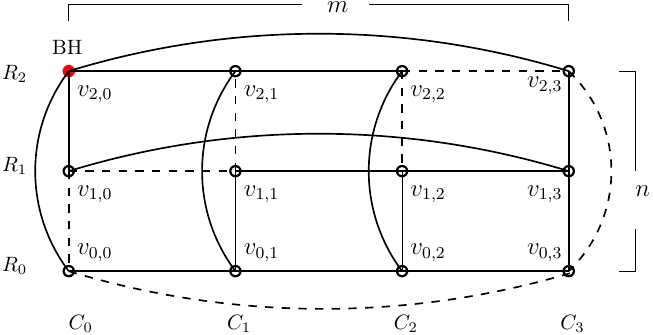}
  	\caption{A $3\times4$ dynamic torus, where dashed edges represents a disappeared or missing edge}
  	\label{fig:example_torus}
  \end{figure}

\subsection{Agent Model}

A set of $k$ agents $A=\{a_1,a_2, \ldots, a_{k}\}$ are assigned the task to locate the black hole in $\mathcal{G}$. We consider two initial configurations in this paper: first, the set of $A$ agents are initially \textit{co-located} at a safe node (the node in $G$ at which they are co-located is termed as \textit{home}), second, the agents are initially \textit{scattered} along safe nodes in $\mathcal{G}$. Each agent in $A$ has a distinct Id of size $\lfloor \log k \rfloor$ bits taken from the set $[1,k]$, and every agent has computational capabilities so that it can communicate with other agents when they are at the same node at the same time. Each agent has knowledge of the torus size, i.e., both $n$ and $m$ are known to the agents. An agent moves from one node to another using the edges at each round; furthermore any number of agents can concurrently move along an edge at any round. These actions are atomic in nature, so an agent cannot recognise the other agents concurrently passing through the same edge at the same round; but it can see and communicate with all the other agents present at the current node at the same round. These agents operate in \textit{synchronous} rounds, so in each round, every agent becomes active and takes a local snapshot of its current node. The snapshot includes the presence of the ports of its current node at the current round, the agent's local memory, the set of agents present at the current node, and the contents of the whiteboard. Now, based on this information the agent performs the following actions:
\begin{itemize}
    \item {\it Look:} In this step, the agent takes the \textit{snapshot} of the current node. This snapshot helps the agent gather the information about the Ids of other agents residing at the same node, the edges that currently exist at the current round and also the whiteboard information at the current node.  
    \item {\it Compute:} On the basis of its earlier snapshot and local memory, the agent decides to stay at the current node or move to another node. The direction of its movement is also calculated in this step.
    \item {\it Move:} In this step if the agent decides to move along a certain direction and if the corresponding edge is present, then it moves along this edge while updating the whiteboard (if required, based on the algorithm) to the next node in the subsequent round.
\end{itemize}

Since, the agents operate in \textit{synchronous} rounds, so each agent gets activated at each round and performs the LCM cycle. So, the time taken by any algorithm is calculated in terms of the \textit{rounds}.

\subsection{Configuration}
A configuration $C_r$ at a round $r$ is defined to be the amalgamation of the presence of the number of agents at a node, the local memory of each agent and contents of the whiteboard at the start of round $r$. The transformation from $C_r$ to $C_{r-1}$ depends on multiple factors, first, the execution of the algorithm, second, the adversarial choices of edges disappeared and reappeared in round $r-1$. $C_0$ is the initial configuration, where, in the co-located case, the initial safe node is chosen by the adversary, whereas in the scattered case, the adversary arbitrarily places the agents along the safe nodes.

The problem of black hole search (or BHS) is defined as follows.
\begin{defn}
    Given a dynamic torus $\mathcal{G}$ of size $n \times m$ ($3\leq n \leq m$), an algorithm $\mathcal{A}$ for a set of $k$ agents solves the BHS problem if at least one agent survives and terminates. The terminating agent must correctly know the exact position of the black hole in the footprint of $\mathcal{G}$. 
\end{defn}
The measures of the complexity for the BHS problem are as follows: the number of agents or \textit{size}, required to successfully execute $\mathcal{A}$, the \textit{time} or the number of rounds required to execute $\mathcal{A}$. Note that in this paper, we have assumed the fact that whenever an agent correctly locates the black hole, the algorithm terminates, so all the other agents executing any action gets terminated immediately.

\section{Lower Bound Results}\label{section-LB}

In this section, we present the lower bound results on the number of agents and number of rounds, in both scenario when the agents are either initially co-located or scattered.

\subsection{Co-located Agents}

The next theorem gives impossibility result on the number of agents when they are initially co-located.

\begin{theorem}\label{theorem:colocatedAgentLB}
    Given a dynamic torus $\mathcal{G}$ of size $n \times m$, there does not exist a BHS algorithm which correctly locates the black hole with $k=n+1$ co-located agents and each node in $\mathcal{G}$ contains a whiteboard of $O(1)$ bits.
\end{theorem}

\begin{proof}
   Suppose $\mathcal{H}$ be any BHS algorithm which works with $k=n+1$ co-located agents in $\mathcal{G}$. Since, in the worst case each node of $\mathcal{G}$ needs to be explored by at least one agent in order to locate the black hole. So, while executing $\mathcal{H}$, whenever an agent visits a node of the form $v_{i,i}$ ($\forall~ 0\leq i \leq n-1$), the adversary has the ability to stop one agent at each such node from moving further in any direction (refer the nodes $v_{0,0},\ldots,v_{2,2}$ in Fig. \ref{fig:example_torus}). It is because, as these nodes are independently located in separate rows and columns, so the adversary has the ability to stop the agent from moving either horizontally or vertically by disappearing either of these edges, which in turn restricts the agent to move any further. In the worst case, $n$ among $n+1$ agents can be stuck in the nodes of the form $v_{i,i}$. This means, if the black hole is not yet detected, and there are lets say $t$ ($>0$) many nodes left to be explored, then this $n+1$-th agent is the only agent to be able to move and hence needs to explore these remaining nodes. While exploring, if this agent visits multiple nodes before reporting, then it is impossible for the agents to correctly locate the black hole. So, the only way this agent can move, is after each new node it explores, it must try to report at least one among the remaining $n$ stuck agents. Now, this also leads to impossibility, because the $n+1$-th agent may encounter a missing edge either while trying to explore a new node or while returning back to report, and it may restrict the agent from reaching its designated location. In this situation, the agent which is waiting for this $n+1$-th agent has no idea whether the agent has indeed entered the black hole or it is stuck due to a missing edge, hence in any situation it cannot ever terminate the algorithm even if the $n+1$-th agent enters the black hole. This shows that it is impossible for $k=n+1$ co-located agents to correctly locate the black hole.\end{proof}

 \begin{blab}\label{corollary:colocatedAgentLB}
     Any BHS algorithm on a dynamic torus $\mathcal{G}$ of size $n\times m$ requires at least $k=n+2$ co-located agents to correctly locate the black hole when each node in $\mathcal{G}$ has a whiteboard of $O(1)$ bits.
 \end{blab}

The next lemma gives a lower bound on the round complexity for any exploration algorithm operating along a dynamic ring, where the agents are initially co-located.
\begin{lemm}\label{lemma:colocated4agentLBRingExploration}
   In a dynamic ring of size $n>3$ in presence of whiteboard, any exploration algorithm with $l$ ($l\geq 2$) co-located agents require at least $\Omega(n)$ rounds to explore a ring of size $n$.
\end{lemm}

\begin{proof}
    Suppose $l$ agents are initially co-located at a node and they execute some exploration algorithm $\mathcal{H}$. Note that any efficient exploration algorithm must instruct the agents to concurrently explore the ring, so the only possibility for concurrency in any such efficient algorithm $\mathcal{H}$, is to instruct some agents to move in a clockwise direction whereas the another set of agents to move in a counter-clockwise direction. In this situation, while executing $\mathcal{H}$ observe that at each round at most 2 nodes can be explored (if none are blocked by a missing edge), so in order to explore a ring of size $n$ at least $\frac{n}{2}=\Omega(n)$ rounds are required.  
\end{proof}

\begin{theorem}[\cite{BHSDynRingLuna}]\label{theorem:colocatedDiLunaRingLB}
    In a dynamic ring of size $n>3$, any BHS algorithm with 3 co-located agents in presence of whiteboard requires $\Omega(n^{1.5})$ rounds, even if the agents have distinct Ids. 
\end{theorem}

The following corollary follows from Lemma \ref{lemma:colocated4agentLBRingExploration} and Theorem \ref{theorem:colocatedDiLunaRingLB}.

\begin{blab}\label{corollary:colocated4agentRingBHSLB}
    In a dynamic ring of size $n>3$, any BHS algorithm with at least 4 co-located agents in presence of whiteboard requires $\Omega(n)$ rounds, even if the agents have distinct Ids.
\end{blab}

The next theorem gives a lower bound on the round complexity for any BHS algorithm operating on a dynamic torus with $k$ co-located agents.
\begin{theorem}\label{theorem:colocatedLBComplexity}
     Any BHS algorithm with $k=n+c$ co-located agents, where $c\in \mathbb{Z^{+}}$ and $c\geq 2$, on a $n\times m$ dynamic torus requires at least $\Omega(m\log n)$ rounds.
\end{theorem}
\begin{proof}
Given a dynamic torus of size $n\times m$ (with $3\leq n \leq m$) and $k=n+c$ agents are initially co-located at a safe node, observe by Corollary \ref{corollary:colocated4agentRingBHSLB}, $l$ (where $l\geq 4$) agents can perform BHS in presence of whiteboard along a row ring of size $m$ in at least $\Omega(m)$ rounds. Now, let us consider there exists an algorithm $\mathcal{H}$ which is tasked to perform BHS along the dynamic torus $\mathcal{G}$, so concurrently exploring a set of rings by a set of $l$ agents is always a better strategy rather than exploring a ring one at a time by a set of agents. Hence, we consider $\mathcal{H}$ instructs a set of $l$ agents to explore a set of rings concurrently. So, if $t$ (where $t\leq \frac{k}{l}$) rings are concurrently explored by the set of $k$ agents, then as each ring in $\mathcal{G}$ is 1-interval connected, so the adversary has the ability to block an agent each in every $t$ such rings (refer the agents $a_1,~a_2,\ldots,~a_6$ in Fig. \ref{fig:initial_configuration}). This means the remaining agents left to explore for the next concurrent exploration is at least $k-\frac{k}{l}$, where each of these concurrent exploration requires $\Omega(m)$ rounds and the number of rings till now explored is $\frac{k}{l}$. In the next concurrent exploration, at least $\frac{k-\frac{k}{l}}{l}$ row rings can be explored in $\Omega(m)$ rounds, which further blocks this many agents, and the remaining agents left to explore remaining graph is $k-\frac{k}{l}-\frac{k-\frac{k}{l}}{l}$, whereas the total number of row rings explored yet is $\frac{k}{l}+\frac{k-\frac{k}{l}}{l}$. Continuing this way, we can define a recursion relation on the remaining number of agents, $T(\alpha)=T(\alpha-1)(1-\frac{1}{l})$, where $T(\alpha)$ resembles that at the $\alpha$-th iteration this many agents are left to explore the remaining part of $\mathcal{G}$ and each such concurrent exploration for black hole requires $\Omega(m)$ rounds. So, for $\alpha$ many iterations $\mathcal{H}$ requires $\alpha\Omega(m)=\Omega(\alpha m)$ rounds. Now, we try to approximate the value of $\alpha$. Observe, when $T(\alpha)\leq 7$, then either the whole torus is explored in the worst case for the black hole or there is no further concurrency possible because in order to concurrently explore at least two rings in $\Omega(m)$ rounds, a minimum of 8 agents (as 4 agents are at least required to explore a ring in $\Omega(m)$ rounds) are required to be left available, so if at most 7 agents are remaining that means no concurrency is possible for any BHS algorithm. Hence, for $T(\alpha) \leq 7$, we approximate the value of $\alpha$. 
     


{
  $ \displaystyle
    \begin{aligned} 
        T(\alpha) \leq 7  \implies T(\alpha-1)\left(1 - \frac{1}{l}\right) \leq 7 \implies T(\alpha-1) \leq \frac{7l}{l-1} \\
        \implies T(\alpha - 2) \left(1 - \frac{1}{l}\right)\leq \frac{7l}{l-1} \leq 7{\left(\frac{l}{l-1}\right)}^2 \cdots \implies T(1)\leq  7{\left(\frac{l}{l-1}\right)}^{\alpha-1}\\
        \implies k\left(1-\frac{1}{l}\right)\leq  7{\left(\frac{l}{l-1}\right)}^{\alpha-1} \implies k \leq 7  {\left(\frac{l}{l-1}\right)}^{\alpha} \implies \frac{\log k - \log 7}{\log \left(\frac{l}{l-1}\right)} \leq \alpha
        \end{aligned}
  $ 
\par}


This implies $\alpha \approx \log n$, as $k=n+c$ and $l\geq 4$. Hence, this means that for any algorithm $\mathcal{H}$, in order to either explore the whole dynamic torus for a black hole or to stop concurrent exploration, at least $\alpha \approx \log n$ many concurrent exploration needs to be performed, where each iteration takes $\Omega(m)$ rounds. This concludes that the total number of rounds at least required by any algorithm with $k=n+c$ co-located agents is $\Omega(m\log n)$.
\end{proof}

\subsection{Scattered Agents}
 \begin{figure}
  	\centering
  	\includegraphics[scale=0.6]{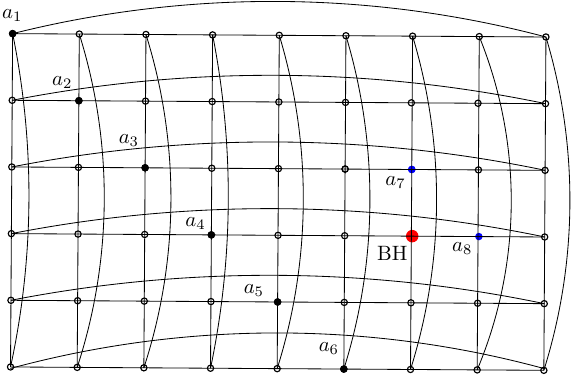}
  	\caption{A initial configuration when the agents are scattered along the dynamic torus}
  	\label{fig:initial_configuration}
  \end{figure}
  The following theorem shows the impossibility to locate the black hole with $k=n+2$ scattered agents. 
\begin{theorem}\label{theorem:scatteredAgentLB}
    Given a dynamic torus $\mathcal{G}$ of size $n\times m$, there does not exist any BHS algorithm which can correctly locate the black hole with $k=n+2$ scattered agents, the result holds as well even if the nodes in $\mathcal{G}$ has a whiteboard.
\end{theorem}
\begin{proof}
Consider a dynamic torus $\mathcal{G}$ of size $n\times m$ (where $m=n+2$), along which the adversary places $n$ among $n+2$ agents at the nodes of the form $v_{i,i}$ (where $0\leq i \leq n-1$), whereas the remaining two agents $a_{n+1}$ and $a_{n+2}$ are initially placed along the nodes of $C_t$ (where $t \in \{n,n+1\}$). Now, as per the placement of first $n$ agents, the adversary has the ability to keep them fixed at their initial position by removing a horizontal and vertical edge, restricting any possible movement (refer the agents $a_1,\ldots,a_6$ in Fig. \ref{fig:initial_configuration}). Based on the movement of the remaining two agents (i.e., $a_{n+1}$ and $a_{n+2}$) while executing some BHS algorithm $\mathcal{H}$, we have the following cases.

\begin{itemize}
    \item Let both the agents move horizontally (resp, vertically) until one among them after $\alpha$ steps, say, first moves vertically (resp, horizontally). Without loss of generality let that agent be $a_{n+2}$ and that node in which it moves after its first vertical (resp, horizontal) move be, $v_{i,j}$. In this situation, the adversary itself places the black hole at $v_{i,j}$ and the other agent $a_{n+1}$ is according placed at either $v_{i,j-1}$ (resp, $v_{i-1,j}$)  or $v_{i,j+1}$ (resp, $v_{i+1,j}$) based on its first horizontal (resp, vertical) move. So, the black hole consumes $a_{n+1}$ after its first step whereas $a_{n+2}$ gets consumed at $\alpha+1$-th step since its execution. In this scenario, both the agents are consumed, whereas the remaining $n$ stuck agents have no knowledge about the black hole position, and cannot ever terminate the algorithm.
    \item If one agent's (say, $a_{n+1}$) first move is along horizontal direction and the other agent's first move is along vertical direction, then also placing the black hole at the junction will invariably consume both these agents, whereas the remaining $n$ agents neither have any knowledge about the black hole position nor they can move from their current position.
\end{itemize}
So, in each possible case we show that with $k=n+2$ agents it is not possible to correctly locate the black hole.\end{proof}

\begin{blab}\label{corollary:scatteredAgentLB}
    Any BHS algorithm on a dynamic torus $\mathcal{G}$ of size $n\times m$ requires at least $k=n+3$ scattered agents to correctly locate the black hole when each node in $\mathcal{G}$ has a whiteboard of $O(1)$ bits.
\end{blab}

Following theorem is inspired from Theorem 4.2 in \cite{DynTorusExpGotoh}, which gives the lower bound on the round complexity for any BHS algorithm with $k$ scattered agents along $\mathcal{G}$.
\begin{theorem}\label{theorem:scatteredRoundLB}
    Any BHS algorithm with $k=n+c$ scattered agents, where $c\in \mathbb{Z^{+}}$ and $c\geq 3$, on a $n\times m$ dynamic torus $\mathcal{G}$ requires at least $\Omega(mn)$ rounds.
\end{theorem}

\begin{proof}
    Let us consider a configuration where the initial position of each $a_i$ is along $v_{i,i}$, $\forall ~0\leq i \leq n-1$ (refer Fig. \ref{fig:initial_configuration}), as these agents are located in separate rows and columns, so the adversary can delete an edge along a row and a column, restricting the agents from moving from its initial position. So, the remaining $c$ agents needs to explore the remaining nodes in order to locate the black hole, and this requires at least $\frac{nm-n}{c}=\Omega(nm)$ rounds.
\end{proof}

\section{Preliminaries}\label{section-preliminaries}

 In this section, we explain all the subroutines, definitions and ideas used in our BHS algorithm, but first, we explain the contents maintained by the agents in the whiteboard.\\

\noindent\textbf{Whiteboard:} The following data is stored and maintained in the whiteboard by the agents. For each $dir \in \{east, west, north, south\}$ with respect to each $v_{i,j}\in \mathcal{G}$ we define the function $f:\{east,west,north,south\}\rightarrow \{\bot,0,1\}$,

\[f(dir) = 
     \begin{cases}
       \text{$\bot$,} &\quad\text{if an agent is yet to visit the port $dir$}\\
       \text{$0$,} &\quad\text{if no agent has marked the port $dir$ as safe}\\
       \text{$1$,}&\quad\text{if the port $dir$ is marked safe}\\
     \end{cases}\]

\noindent\textbf{Cautious Walk:} This is a fundamental movement strategy used in a network with black hole and it is used as a building block of all our algorithms. In this strategy, if two agents are together, then this strategy ensures that only one among them enters the black hole, while the other survives. On the contrary, if only a single agent is present, then whenever it visits a new node, it leaves some mark behind in the whiteboard, so that whenever another agent tries to visit this node along the same edge, it finds the mark and does not enter the black hole. 

This walk is performed in three rounds, where if an agent $a_1$ (say) is alone (resp, with another agent $a_2$, say) then in the first round $a_1$ decides to move one step along $e=(u,v)$ from $u$ to $v$ by marking $f(e)=0$ (while $a_2$ waits) and if it is safe, i.e., does not contain the black hole, then in the next round $a_1$ returns to $u$ and marks the edge $e$ safe by writing $f(e)=1$, then in the third step $a_1$ (resp, $a_2$) moves to $v$. This strategy ensures that no two agent enters the black hole along the edge $e$.\\

\noindent\textbf{Stuck:} An agent $a_1$ is defined to be \textit{stuck} while exploring a 1-interval connected ring for two reasons. 
    \begin{itemize}
        \item First, if while performing \textit{cautious} walk along an edge $e=(u,v)$, $a_1$ at round $r$ marks $f(e)=0$ at $u$ and moves to $v$, while $v$ is safe and $a_1$ tries to return to $u$ at round $r+1$ to mark $f(e)=1$, finds $e$ to be missing, in this situation $a_1$ is \textit{stuck} at $v$ until $e$ reappears.
        \item Second, if while moving along $dir$, $a_1$ finds $e$ to be missing. In this situation, if more than one agent is simultaneously trying to move along $dir$ at the same round and if $a_1$ is the lowest Id among them, then $a_1$ is \textit{stuck} until $e$ reappears, or, if $a_1$ is alone, then in that case also $a_1$ is \textit{stuck} until $e$ reappears.
    \end{itemize}


\subsection{Subroutines}
In this section, we will discuss the sub-routines used as a building block in our BHS algorithms for both the co-located and scattered initial configurations. We have followed some of the pseudocode convention from the papers \cite{BHSDynRingLuna} and \cite{DynTorusExpGotoh}. In this paper, we use three kinds of \textsc{Move} procedure in our algorithms, first, \textsc{Move}$(dir ~|~p_1:s_1;p_2:s_2;\ldots ; p_k:s_k)$, second, \textsc{Move}$(dir \rightarrow f(dir)~|~p_1:s_1;p_2:s_2;\ldots ; p_k:s_k)$, and lastly, \textsc{Move}$(dir\rightarrow f(dir) \rightarrow s_i~|~p_1:s_1;p_2:s_2;\ldots ; p_k:s_k)$,  where $p_i$ is the predicate corresponding to the state $s_i$ and $f(dir)$ represents the value with respect to $dir$ (where $dir\in \{east, west, north,south\}$) in the whiteboard, so depending on the algorithm we use either of these \textsc{Move} procedures. The agent at each round, first takes a snapshot at its current location, and thereafter checks the predicates $p_1,\ldots,p_k$ one after another. If no predicate is satisfied, then in the first \textsc{Move} procedure, the agent moves along the direction $dir$, in the second \textsc{Move} procedure the agent moves along $dir$ while updating the whiteboard of the current node along $dir$ to $f(dir)$, and lastly, in the third \textsc{Move} procedure, in addition to moving along $dir$ and updating the whiteboard, it also moves directly in to the state $s_i$. On the otherhand, if some predicates are satisfied, then the agent chooses the first satisfied predicate (say) $p_i$, and the procedure stops, and the agent moves in to state $s_i$ corresponding to $p_i$. The predicate and state of the form $p_j:time+i \rightarrow s_j$ indicates that if $p_j$ is satisfied then the agent enters the state $s_j$ after $time+i$ rounds, whereas the predicate and the state of the form $p_j:f(dir)\rightarrow s_j$, indicates that if $p_j$ is satisfied then the agent performs the action $f(dir)$ and then moves to the state $s_j$. Further, all this procedure is again executed in the subsequent rounds. The list of all the basic predicates are explained in table \ref{table:explanation-predicates} whose compositions are used as predicates in our algorithms.

\begin{table}

\begin{tabular}{|c|c|}
  \hline
  Basic Predicates & Explanation \\
  \hline
  $time$ & The number of rounds since the start of the algorithm \\
  \hline
  $read[f(dir)]$ & Represents the data read by the agent \\
  & in whiteboard along $dir$ of current node\\
  \hline
  $MEdir$ & Indicates the edge along $dir$ is missing\\
  \hline
  $\overline{ME}dir$ & Indicates the edge along $dir$ exists \\
  \hline
  $catches$ & Indicates that the agent finds another agent, either moving along\\ 
  & the same direction or stuck at the same node \\
  \hline
  $catches[i]$ & Implies that the agent $catches$ another agent with Id $i$\\
  \hline
  $catches-waiting$ & Represents that the agent $catches$ another agent waiting \\
  \hline
  $lowestId$ & Indicates the agent is having lowest Id in its current node \\
  \hline
  $\neg lowestId$ & Indicates the agent is not the lowest Id in the current node \\
  \hline
  $Enodes$ & Stores the number of nodes traversed since the last call of \textsc{Move}\\
  \hline
\end{tabular}
\caption{Explains the list of basic predicates}
\label{table:explanation-predicates}
\end{table}



In the following part we define the algorithm \textsc{Cautious-WaitMoveWest}().\\

\noindent\underline{\textsc{Cautious-WaitMoveWest}$(j,l)$}: This algorithm works on 1-interval connected ring $R_i$ (say), where the main purpose is to make a certain number of agents reach the node $v_{i,j}$ along the $C_j$-th column from any initial configuration. Further, whenever an agent reaches the desired node and it is not stuck, it waits at that node until further instruction is provided.

The algorithm works as follows: for the first $4(l-1)m$ rounds, if an agent $a_1$ is instructed to perform \textsc{Cautious-WaitMoveWest}$(j,l)$ along $R_i$, then it starts the following procedure, if the agent $a_1$ (say) is initially with another agent $a_2$ (say) and since $a_1$ is the lowest Id among them, $a_1$ starts cautious walk along $west$ until it either gets stuck or has reached the desired node. On the other hand the task of $a_2$ is to follow $a_1$ until $a_1$ is stuck. While $a_1$ is stuck, $a_2$ performs the following action:
\begin{itemize}
    \item If $a_1$ is stuck due to a missing edge along $west$, then $a_2$ instead of waiting reverses its direction to $east$ and continues to perform \textit{cautious} walk.
    \item If $a_1$ is stuck while returning back to mark a port safe along $west$ which it has in the last round marked unsafe while exploring and, then $a_2$ waits for at most $3m$ rounds since the round it encountered this situation, and then reverses its direction and continues to perform \textit{cautious} walk. 
\end{itemize}
On the otherhand, if $a_1$ is alone, then it performs \textit{cautious} walk until it either reaches the desired node or it is stuck. If $a_1$ catches another agent stuck, and if it is not the lowest Id among them, then it performs the similar action, as explained earlier in case of $a_2$.

After $4(l-1)m$ rounds has passed, each agent not stuck due to a missing edge tries to reach the node $v_{i,j}$. 

 The pseudocode of \textsc{Cautious-WaitMoveWest}$(j,l)$ is explained in Algorithm \ref{alg-1}. We have used 11 normal states of the form $s_i$ and 9 negation states of the form $\neg s_i$. A negation state $\neg s_i$ is the opposite of the normal state $s_i$. More precisely, if $p_i$ is the predicate for $s_i$, then $\neg p_i$ is the predicate for the state $\neg s_i$. For example, the negation of the predicate $read[f(west)=1]\wedge \overline{ME}west\wedge LowestId$ will be $read[f(east)=1]\wedge \overline{ME}east\wedge LowestId$, i.e., negation of a predicate happens only in terms of direction, so if the predicate contains a basic predicate of the form of $lowestId$ or $\neg lowestId$, then they remain same in $\neg p_i$ as well. On the other hand the negation of a state is explained with the help of this example, there exists a state \textit{Init'} with \textsc{Move}$(west|Enodes>0:\textbf{Init})$ in Algorithm \ref{alg-1}, so the corresponding $\neg$\textit{Init'} will be \textsc{Move}$(east|Enodes>0:\neg\textbf{Init})$. So, while the negation of a predicate means only the change in direction, the negation of a state on the contrary means $\neg p_i:\neg s_i$, i.e., both negation of a predicate as well as negation of the corresponding state. In the following part, we give a detailed explanation of some of the states in Algorithm \ref{alg-1}. 

\begin{itemize}
    \item \textbf{Init}: resembles first step of \textit{cautious} walk, i.e., the first round of \textit{cautious} walk when an agent tries to explore an unexplored node. In this situation, the agent performs the action $west \rightarrow f(west)=0 \rightarrow \text{Backtrack}^{0}$, i.e., if no predicate is satisfied, then the agent moves along $west$ by updating $f(west)=0$ from $\bot$ and then enters in to state $\textit{Backtrack}^{0}$.
    \item \textbf{Backtrack$^{0}$}: signifies the second step of \textit{cautious} walk, where the agent performs $east \rightarrow f(east)=1 \rightarrow \text{Init}$, i.e., if no predicate is satisfied then the agent moves along $east$ while marking $f(east)=1$ and enters the state \textit{Init}.
    \item \textbf{Init''}: this state signifies the last step of cautious walk, where the agent performs $west \rightarrow f(west)=1 \rightarrow \text{Init}$, i.e., if no predicate is satisfied then the agent moves along $west$ to the new node while updating $f(west)=1$ and enters \textit{Init} state.
    \item \textbf{Init'}: this state signifies the fact that the agent is instructed to move along $west$ without performing any other action, when it finds that the edge along $west$ exists and it is marked safe.
    \item \textbf{Wait$^{01}$}: an agent enters this state when it finds a forward missing edge along $west$, which is either marked as $\bot$ or $1$ and it is the lowest Id agent at the current node. In this case, the state instructs the agent to wait until the missing edge reappears. 
    \item \textbf{Wait$^{02}$}: this state instructs the lowest Id agent to wait for at most $3m$ rounds, if while it is waiting the predicate $time\geq g(l)$ ($g(l)=4(l-1)m$ refer Algorithm \ref{alg-1}) is satisfied then it enters the state \textit{Return}, or if it has already waited for $3m$ rounds then it enters the state $\neg\textit{Init}$, or if the missing edge reappears while it is waiting then it enters the state $\textit{Wait}^{04}$.
    \item \textbf{Wait$^{03}$}: this state instructs the agent to stay at the current node for one round, and after one round, if the edge along $west$ still exists and no agent returns to mark it safe, then it enters the state \textit{TerminateW}, which terminates the algorithm by declaring that the next node along $west$ is the black hole. Otherwise, if an agent returns to mark it safe or the edge again goes missing then in either case it enters the state \textit{Init}.
    \item \textbf{Wait$^{04}$}: instructs the agent to wait until the edge reappears and whenever the edge reappears the state instructs the agent to enter the state \textit{Backtrack}$^{0}$, in which it directly performs the action $east \rightarrow f(east)=1 \rightarrow \text{Init}$ as instructed in state \textit{Backtrack}$^{0}$.
    \item \textbf{Wait}: signifies the state when the agent has reached its desired node and it is not stuck, in this state the agent is instructed to wait until further instruction.
    \item \textbf{Return}: this state arises when $time>g(l)$, and in this state the agents which are waiting and are not stuck are instructed to move and reach the desired node, i.e., $v_{i,j}$ for at most $3m$ rounds. After $time>g(l)+3m$, only if the agent is stuck while in state \textit{Backtrack}$^0$ and finds the edge reappear then it backtracks and marks the corresponding edge safe and moves to the next node. Otherwise, all the other agents remain at their current position until they receive further instruction. This state signifies the end of the algorithm.
\end{itemize}
Note that the pseudocode in Algorithm \ref{alg-1} explains only for $west$ direction, \textsc{Cautious-WaitMoveSouth}$(j,l)$ is similar, only $east$ and $west$ are replaced by $south$ and $north$, respectively. Also note that, all the lemmas, theorems and corollaries are explained for \textsc{Cautious-WaitMoveWest}, but they hold for other directions as well.

\begin{algorithm2e}[!ht]\footnotesize
\caption{\sc{Cautious-WaitMoveWest}$(j,l)$}\label{alg-1}
States:\{\textbf{Init, Init', Init''  $\neg$Init, $\neg$Init', $\neg$Init'' $\text{Backtrack}^{0}$, $\neg \text{Backtrack}^{0}$, $\text{Wait}^{01}$, $\text{Wait}^{02}$, $\text{Wait}^{03}$, $\text{Wait}^{04}$, $\neg \text{Wait}^{01}$, $\neg \text{Wait}^{02}$, $\neg \text{Wait}^{03}$, $\neg \text{Wait}^{04}$, Wait, Return, TerminateW, $\neg$ TerminateW }\}.\\
$time$ is defined as the number of rounds since the start of the current algorithm.\\
Define $g(l)=4(l-1)m$\\
In state \textbf{Init}: \\
\textsc{Move}$(west \rightarrow f(west)=0\rightarrow \text{\textbf{Backtrack}}^{0}~|~time\geq g(l): \textbf{Return};current\in C_j:\textbf{Wait};read[f(west)=1]\wedge \overline{ME}west:\textbf{Init'};read[f(west)=\bot]\wedge \overline{ME}west\wedge \neg LowestId:time+1\rightarrow\textbf{Init};read[f(west)=\bot \lor f(west)=0 \lor f(west)=1]\wedge MEwest\wedge catches-waiting:\neg\textbf{Init};read[f(west)=1 \lor f(west)=0 \lor f(west)=\bot]\wedge MEwest\wedge \neg LowestId:\neg\textbf{ Init};read[f(west)=1 \lor f(west)=\bot]\wedge MEwest\wedge LowestId:\text{\textbf{Wait}}^{01};read[f(west)=0]\wedge MEwest \wedge LowestId:\text{\textbf{Wait}}^{02}; read[f(west)=0]\wedge \overline{ME}west:\text{\textbf{Wait}}^{03})$\\

In state \textbf{$\neg$Init}: \\
\textsc{Move}$(east \rightarrow f(east)=1\rightarrow \neg \text{\textbf{Backtrack}}^{0}~|~time\geq g(l); \textbf{Exit};current\in C_j:\textbf{Wait};read[f(east)=1]\wedge \overline{ME}east:\neg\textbf{Init'};read[f(east)=\bot]\wedge \overline{ME}east\wedge \neg LowestId:time+1\rightarrow \neg\textbf{Init};read[f(east)=\bot \lor f(east)=0 \lor f(east)=1]\wedge MEeast\wedge catches-waiting:\textbf{Init};read[f(east)=\bot \lor f(east)=0 \lor f(east)=1]\wedge MEeast\wedge \neg LowestId:\textbf{Init};read[f(east)=\bot \lor f(east)=1]\wedge MEeast\wedge LowestId:\neg\text{\textbf{Wait}}^{01};read[f(east)=0]\wedge MEeast \wedge LowestId :\neg\text{\textbf{Wait}}^{02}; read[f(east)=0]\wedge \overline{ME}east:\neg\text{\textbf{Wait}}^{03})$\\

In state \textbf{Backtrack$^{0}$}: \textsc{Move}$(east \rightarrow f(east)=1 \rightarrow \textbf{Init''}|MEeast:f(east)=1\rightarrow\text{\textbf{Wait}}^{04})$\\

In state \textbf{Init'}:\textsc{Move}$(west|Enodes>0:\textbf{Init})$\\

In state \textbf{Init''}: \textsc{Move}$(west\rightarrow f(west)=1\rightarrow \textbf{Init}~|~Enodes>0:\textbf{Init}; MEwest:f(west)=1\rightarrow \textbf{Init}$)\\

In state \textbf{Wait$^{01}$}: \textsc{Move}$(nil|\overline{ME}west:\textbf{Init})$\\

In state \textbf{Wait$^{02}$}:\\
\If{$time\mod 3m =0$}
{
\textsc{Move}$(nil|time\geq g(l): \textbf{Return}; time>time+3m:\neg\textbf{ Init};\overline{ME}west:\textbf{Wait}^{03})$\\
}
\Else
{
\textsc{Move}$(nil|time\geq g(l): \textbf{Return}; time \mod 3m=0:\neg\textbf{ Init};\overline{ME}west:\textbf{Wait}^{03})$\\

}
In state \textbf{Wait$^{03}$}:\\
\Comment{$time2$ is the number of rounds since the agent has encountered an edge along $west$ with $f(west)=0$}\\
\textsc{Move}$(nil| time2>1 \wedge read[f(west)=0] \wedge \overline{ME}west: \textbf{TerminateW}; time2>1 \wedge read[f(west)=0] \wedge MEwest: \textbf{Init}; time2>1 \wedge read[f(west)=1]: \textbf{Init})$\\

In state \textbf{Wait$^{04}$}: \textsc{Move}$(nil|\overline{ME}east:\text{\textbf{Backtrack}}^{0})$\\

In state \textbf{Wait}: Wait at the current node, until further instructed.\\

In state \textbf{Return}:\\
\While{$time\leq g(l)+3m$}
{
\If{an agent is waiting for another agent}
{
Change its direction of movement, until $current\in C_j$, while moving if it encounters an edge marked with $\bot$ then move cautiously. If while moving encounters a missing edge, wait until it reappears.\\ 
}
\ElseIf{more than one agent in current node}
{
Lowest Id agent follows current direction, whereas remaining agents change direction and moves until $current\in C_j$ and if it encounters an edge marked with $\bot$ then move cautiously. If while moving encounters a missing edge, wait until it reappears.\\ 
}
}
If last in state \textbf{Backtrack}$^0$ or $\neg$\textbf{Backtrack}$^0$ and the edge reappears then perform one step of backtrack, mark the port 1 and return to the previous node. Otherwise, remain at the current node until further instruction.\\
In state \textbf{TerminateW}: Terminate, BH is next node towards $west$.\\
\end{algorithm2e}

\begin{obs}[\cite{BHSDynRingLuna}]
    Given a dynamic ring $R$ and a cut $U$, where $|U|>1$, if its footprint is connected by edges $e_c$ and $e_{cc}$ (where $e_c$ and $e_{cc}$ are the clockwise and counter-clockwise edges, respectively) to nodes in $V\backslash U$ (where $V$ is the set of vertices not in $U$). If all the agents at a round $r$ are at $U$, does not try to cross along $e_c$, whereas there exists an agent which tries to cross along $e_{cc}$, then the adversary has the ability to prevent any agent from crossing $U$.
\end{obs}
Our algorithm \textsc{Cautious-WaitMoveWest}() ensures that this situation does not arise, as when an agent is stuck on $e_c$ (or $e_{cc}$), another agent after finding this situation waits for at most $3m$ rounds (depending on the fact that whether the earlier agent is stuck while backtracking or it is stuck because it has encountered a missing edge along $west$), and then reverses its movement towards $e_{cc}$ (or $e_c$), while the other agent remains stuck. Hence, there exists a round $r$ where an agent each is trying to cross $e_c$ and another agent is trying to cross $e_{cc}$.

\begin{lemm}
    If an agent executing Algorithm \ref{alg-1} terminates while moving along a certain direction, then it correctly locates the black hole.
\end{lemm}
\begin{proof}
    Suppose an agent $a_i$ (say) terminates while moving along $west$ (i.e., enters the state $TerminateW$), but the next node along $west$ is not the black hole. The reason $a_i$ has terminated because it has found $read[f(west)=0]$ (i.e., the edge along $west$ is marked unsafe by some other agent) and $\overline{ME}west$ (i.e., the corresponding edge along $west$ exists), so in which case after encountering this situation the agent directly enters state \textit{Wait $^{04}$}. In state \textit{Wait $^{04}$}, the agent is instructed to wait for one round, after which it still finds $f(west)=0$ and the edge exists, i.e., no agent has returned and marked $f(west)=1$ while the edge still exists. This guarantees the next node along $west$ to be the black hole because, according to the principle of \textit{cautious} walk, the agent which at round $r$ (say), first traversed along this edge to mark it unsafe (i.e., by writing $f(west)=0$) must eventually return to mark it safe whenever the corresponding edge exists (refer states \textit{Init} and $\textit{Backtrack}^0$), if not already consumed by the black hole. So, our algorithm correctly locates the black hole.
\end{proof}
\begin{blab}
    \textsc{Cautious-WaitMoveWest}() ensures that at most two agents enters the black hole.
\end{blab}
\begin{lemm}\label{lemma:CorrAlg1}
    If two agents along a safe row ring $R_i$ of size $m$ ($m\geq 3$) executes Algorithm \ref{alg-1}, then at least one agent reaches the desired node within $7m$ rounds.
\end{lemm}
\begin{proof}
    Suppose, without loss of generality, $a_1$ and $a_2$ are initially co-located at a node $v_{i,j+1}$ along $R_i$ and they start executing \textsc{Cautious-WaitMoveWest}$(j,2)$. We define the worst scenario, when in the first round, $a_1$ moves one step towards $west$, i.e., to the the node $v_{i,j}$, while updating $f(west)=0$ at $v_{i,j+1}$. In this round, the other agent $a_2$ after finding that it is not the lowest Id and $f(west)=\bot$ waits for one round. In the next round, suppose the adversary disappears the edge $(v_{i,j}, v_{i,j+1})$, so $a_1$ after finding this, remains stuck at $v_{i,j}$ in state $\textit{Backtrack}^0$ until the edge reappears. On the other hand, $a_2$ finds that in the second round $f(west)=0$ and the edge along $west$ is also missing, so it waits till $time=3m$ rounds at $v_{i,j+1}$, and still if the edge remains missing, then it enters state \textit{Return}, in which since $a_1$ is stuck but $a_2$ is not stuck and waiting for $a_1$, so $a_2$ changes its direction towards $east$ and starts moving \textit{cautiously}, as the nodes along $east$ are marked with $\bot$. Now, if $a_2$ finds a missing edge along $east$ at round $r$ (where $r<7m$), then $a_1$ finds the edge reappeared, so at round $r$ it returns to $v_{i,j+1}$ marks $f(west)=1$ and at round $r+1$ reaches $v_{i,j}$. Otherwise, if $a_2$ never encounters a missing edge, then along $east$ in at most $3m$ rounds it reaches the desired node $v_{i,j}$ and the ring is also explored, which proves our claim. \end{proof}

\begin{lemm}\label{lemma:Impossibile3agentAlg1}
    If three agents are executing \textsc{Cautious-WaitMoveWest}$(j,3)$ along $R_i$ and $v_{i,j}$ is the black hole, then at most 2 agents enter the black hole whereas the adversary has the ability to stop the third agent from detecting the black hole location.
\end{lemm}
\begin{proof}
    Suppose at round $r$ (where $r\leq 3m$) an agent $a_1$ enters the black hole node $v_{i,j}$ from $v_{i,j+1}$, then suppose at round $r+t$ (where $t\in\mathbb{Z}^+$ and $t>0$) another agent $a_2$ reaches the adjacent node $v_{i,j+1}$, only to find $(v_{i,j},v_{i,j+1})$ is missing and $f(west)=0$. In this situation, $a_2$ waits at most $3m$ rounds and then moves towards $east$ while it enters the state $\neg\textit{Init}$. In the meantime, while $a_2$ is already waiting at $v_{i,j+1}$ for the first $3m$ rounds, $a_3$ at some point $catches$ $a_2$, then finding the $a_2$ is already waiting, $a_3$ immediately moves towards $east$ by entering the state $\neg\textit{Init}$ (as the predicate $read[f(west)=\bot \lor f(west)=0 \lor f(west)=1]\wedge MEwest \wedge catches-waiting$ in state \textit{Init} is satisfied) and in at most $3m$ rounds enters the black hole node $v_{i,j}$ from $v_{i,j-1}$ along $east$. Next, whenever $a_2$ reaches $v_{i,j-1}$, adversary again disappears $(v_{i,j-1},v_{i,j})$ and whenever it reaches $v_{i,j+1}$, adversary reappears $(v_{i,j-1},v_{i,j})$ and disappears $(v_{i,j},v_{i,j+1})$, restricting $a_2$ to locate the black hole, whereas both $a_1$ and $a_3$ has been consumed by the black hole.\end{proof}

    \begin{blab}\label{corollary:Possibile4agentAlg1}
        A set of 4 agents, executing \textsc{Cautious-WaitMoveWest}$(j,l)$ (where $l\geq 4$) along $R_i$ can correctly locate the black hole in at most $15m$ rounds, where $v_{i,j}$ is the black hole node. 
    \end{blab}

 \begin{blab}\label{corollary:2agentsBHAlg1}
     \textsc{Cautious-WaitMoveWest}$(j,l)$ ensures that exactly 2 agents can be consumed by the black hole when the desired node $v_{i,j}$ is also the black hole node.
     \end{blab}

\begin{figure}
  	\centering
  	\includegraphics[scale=0.6]{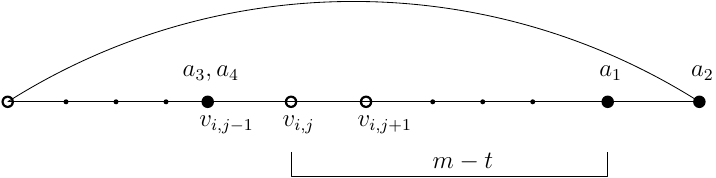}
  	\caption{An instance where 4 agents are operating along a ring while executing Algorithm \ref{alg-1}}
  	\label{fig:partition}
  \end{figure}

\begin{lemm}\label{lemma:4mRound1agentDesiredNodeAlg1}
    After the first $4m$ rounds has elapsed executing \textsc{Cautious-WaitMoveWest}$(j,l)$ (where $l>2$), if at least 3 agents are still present along $R_i$ then it takes at most $4m$ rounds for an agent among them to reach the desired node, since the last agent has reached the desired node.
\end{lemm}

\begin{proof}
    Suppose at round $r-1$ there exists two agents $a_3$ and $a_4$ which encounters a missing edge along $east$ at $v_{i,j-1}$ before reaching the desired node $v_{i,j}$. In this situation $a_3$ satisfies the predicate $read[f(east)=1 \lor f(east)=\bot]\wedge MEeast \wedge LowestId$ in state $\neg Init$ hence it waits at the node $v_{i,j-1}$, whereas $a_4$ satisfies the predicate $read[f(east)=\bot \lor f(east)=0 \lor f(east)=1]\wedge MEeast \wedge \neg LowestId$ and hence moves in to state \textit{Init} by reversing its direction to $west$. On the otherhand, $a_1$ and $a_2$ be the other two agents which are currently trying to reach $v_{i,j}$ while walking \textit{cautiously} from $west$. Next, in round $r$ suppose the adversary reappears the edge $(v_{i,j-1},v_{i,j})$ whereas disappears the edge between $a_1$ and $a_2$, which leads $a_3$ to reach the desired node in at most 3 rounds (if $f(east)=\bot$), whereas $a_4$ has already started to move $west$ while in state $Init$. Now, at this point suppose there are $m-t$ (where $t\in \mathbb{Z}^+$ and $1<t<m$) nodes left to be explored along $west$ from the current position of $a_1$ to the desired node (refer Fig. \ref{fig:partition}). So, in at least $t$ rounds $a_4$ must reach $a_2$ which is currently waiting for at most $3m$ as the predicate $read[f(west)=0]\wedge MEwest \wedge LowestId$ in state $Init$ is satisfied. The moment $a_4$ reaches the node containing $a_2$, adversary reappears the edge between $a_1$ and $a_2$, so $a_2$ which is in state $\textit{Bactrack}^0$ returns to $a_1$ and marks the corresponding port to 1. So after this all three agents starts to move cautiously along $west$ for $3(m-t)-2$ rounds, i.e., the point at which $a_1$ first visits the desired node from $v_{i,j+1}$, and the remaining agents $a_2$ and $a_4$ are at $v_{i,j+1}$. At this point, adversary again disappears the edge $(v_{i,j},v_{i,j+1})$, so this means $a_1$ remains stuck in state $\textit{Backtrack}^0$ until the edge reappears, $a_2$ starts waiting for at most $3m$ rounds, but $a_4$ satisfying the predicate $read[f(west)=1 \lor f(west)=0 \lor f(west)=\bot]\wedge MEwest \wedge LowestId$ reverses its direction of movement by moving in to state $\neg Init$. So, in another $m-t-1+t$ steps, it reaches $v_{i,j-1}$, i.e., the other side of $v_{i,j}$ and in the next 2 rounds, at least one of the agents between $a_1$, $a_2$ and $a_4$ reaches the desired node. So, the total number of rounds required since $a_3$ has reached the desired node is: $t+3(m-t)-2+m-1+2=4m-2t-1\leq 4m$.  
\end{proof}

The following corollary follows from Lemmas \ref{lemma:CorrAlg1} and \ref{lemma:4mRound1agentDesiredNodeAlg1}.  
\begin{blab}\label{corollary:4(l-2)mDesiredNodeAlg1}
    Our algorithm ensures that among $l$ agents operating along $R_i$ at least $l-2$ agents reach the desired node within $4(l-1)m$ rounds.
\end{blab}
     
\begin{lemm}\label{lemma:StateReturnAlg1}
    Among the remaining two agents which enter state \textit{Return} after $4(l-1)m$ rounds has elapsed, at least one among them reaches the desired node.
\end{lemm}

\begin{proof}
    Suppose two agents $a_1$ and $a_2$ enter state \textit{Return}, then there are the following scenarios.
    \begin{itemize}
        \item Without loss of generality, if $a_1$ is trying to move along $west$ and $a_2$ is trying to move along $east$, then since the adversary can disappear at most one edge at any round, so either of the agent reaches the desired node in at most $3m$ rounds.
        \item If both the agents move along the same direction, then also we have following scenarios.
        \begin{itemize}
            \item If without loss of generality, while both agents enter state \textit{Return}, $a_1$ and $a_2$ are along the adjacent nodes, where $a_2$ is waiting for $a_1$ to return, whereas $a_1$ in state $\textit{Bactrack}^0$ is waiting for the missing edge to reappear. In this scenario, our algorithm instructs the waiting agent to change its direction, and move until either $3m$ rounds has passed, or the agent reaches its desired node. In this case also, since both the agents start moving along opposite direction, so either of them reaches the desired node within $3m$ rounds.
            \item If both the agents are together and moving in same direction, in this situation, the algorithm instructs the lowest Id agent, i.e., $a_1$ to continue moving along its direction, whereas $a_2$ must reverse its direction. So, by the earlier argument again, one among these two agents reach the desired node in at most $3m$ rounds.
        \end{itemize}
    \end{itemize}
    So, in each case, we show that in at most $3m$ rounds, at least one among the remaining two agents reach the desired node.
\end{proof}

The following theorem follows from Corollary \ref{corollary:4(l-2)mDesiredNodeAlg1} and Lemma \ref{lemma:StateReturnAlg1}.
\begin{theorem}\label{theorem:corrAlg1}
    If $l$ agents ($l\ge 2$) agents are in a safe ring $R_i$ and they perform \textsc{Cautious-WaitMoveWest}$(j,l)$, then at least $l-1$ agents reach and stay on $v_{i,j}$ within $4(l-1)m+3m$ rounds, since the start of execution of Algorithm \ref{alg-1}.
    \end{theorem}

 \noindent\underline{\textsc{Cautious-Move}}$(west,j)$: This algorithm is a special version of Algorithm \ref{alg-1}, it has two stages, and requires at least 2 agent. The first stage is exploration and works for $3m$ rounds, and the second stage is Exit. The algorithm works as follows, the lowest Id agent becomes the $Leader$ whereas the other agents becomes the $Follower$. The $Leader$ follows Algorithm \ref{alg-1a}, whereas the $Follower$ follows Algorithm \ref{alg-1b}. The $Leader$ explores new nodes in first stage and $Follower$ follows the $Leader$ until it either finds the $Leader$ to be stuck or $Leader$ stops reporting either due to a missing edge or it has entered the black hole. Whenever, the $Follower$ finds the edge is missing and $Leader$ is not reporting and $time<3m$ it waits until the missing edge reappears or till $time=3m$, whereas if it finds that the edge exists and $Leader$ is not reporting then it terminates the algorithm, by declaring the node in which $Leader$ has explored is the black hole node. On the other hand, whenever the $Leader$ is also stuck due to a missing edge along its moving direction, then both $Leader$ and $Follower$ waits until $time=3m$. Whenever $time>3m$, both $Leader$ and $Follower$ enter the second stage, i.e., state \textit{Exit}, in which, irrespective of the fact that they are stuck or not, they try to return to their desired node, i.e., the node along $C_j$-th column, while returning each agent irrespective of $Leader$ or $Follower$ is instructed to mark the port of each node along their movement to 1 if not already marked so. Whenever the agents while returning back encounters a missing edge, the lowest Id agent waits and other agents changes direction.  

 The pseudocode of \textsc{Cautious-Move}$(west,j)$ is explained in algorithm \ref{alg-1a} and \ref{alg-1b}. Also, note that the algorithm is similar for other directions as well only change happens in the direction. The lemmas and corollaries explained for \textsc{Cautious-Move}$(west,j)$, also holds for other directions as well.

\begin{algorithm2e}\footnotesize
 \caption{\textsc{Cautious-MoveLeader}$(west,j)$}\label{alg-1a}
 In state \textbf{NewNode}:\\
 \textsc{Move}$(west\rightarrow f(west)=0\rightarrow \textbf{Backtrack}|~time>3m:\textbf{Exit};MEwest: \textbf{Wait};\overline{ME}west \wedge read[f(west)=1]:\textbf{Init};\overline{ME}west \wedge read[f(west)=0]:\textbf{Wait}^{01})$.\\
 In state \textbf{Backtrack}: \textsc{Move}$(east \rightarrow f(east)=1 \rightarrow \textbf{Next}|~time>3m:\textbf{Exit};MEeast:f(east)=1\rightarrow\textbf{Wait}^{02})$\\
 In state \textbf{Next}: \textsc{Move}$(west\rightarrow f(west)=1 \rightarrow \textbf{NewNode}|~time>3m:\textbf{Exit};MEwest:f(west)=1\rightarrow\textbf{Wait}^{03})$\\
 In state \textbf{Init}: \textsc{Move}$(west\rightarrow \textbf{NewNode}~|~time>3m:\textbf{Exit}; MEwest:\textbf{Wait})$\\
 In state \textbf{Wait}:\textsc{Move}$(nil|~time>3m:\textbf{Exit};\overline{ME}west \wedge read[f(west)=1]:\textbf{Init};~\overline{ME}west \wedge read[f(west)=0]:\textbf{Wait}^{01}$)\\
 In state \textbf{Wait}$^{01}$:\\
 \Comment{$time1$ is defined as the number of rounds since the agent has encountered an edge along $west$ with $f(west)=0$.}\\
 \textsc{Move}$(nil|~time1>1 \wedge read[f(west)=0]\wedge \overline{ME}west:\textbf{TerminateW};MEwest: \textbf{Wait};time>3m:\textbf{Exit}$)\\
 In state \textbf{Wait}$^{02}$:\textsc{Move}$(nil|~time>3m:\textbf{Exit};\overline{ME}east:\textbf{Backtrack})$\\
 In state \textbf{Wait}$^{03}$:\textsc{Move}$(nil|~time>3m:\textbf{Exit};\overline{ME}west :\textbf{Next})$\\
 In state \textbf{TerminateW}:Terminate, BH is the next node towards $west$.\\
 In state \textbf{Exit}:\\
 Each agent moves until reaches the node in $j$-th column, if a node is marked $\bot$ or $0$ , mark it 1 and continue its movement.\\
 \If{Encounters a Missing Edge}
 {
 If no agent is already stuck, then lowest Id agent gets stuck for the edge to reappear, others change direction.\\
 }
 \end{algorithm2e}

 \begin{algorithm2e}\footnotesize
 \caption{\textsc{Cautious-MoveFollower}$(west,j)$}\label{alg-1b}
 In state \textbf{Follow}:\\
 \textsc{Move}$(west |~time>3m:\textbf{Exit}; read[f(west)=0]\wedge catches[Leader]:\textbf{Wait}^{01};~read[f(west)=0]\wedge \neg catches[Leader]\wedge \overline{ME}west:\textbf{Wait}; MEwest:\textbf{Wait}^{01};\overline{ME}west\wedge catches[Leader]\wedge read[f(west)=0]:\textbf{Wait}^{01})$\\
 In state \textbf{Wait}$^{01}$:\\
 \textsc{Move}$(nil|~time>3m:\textbf{Exit};\overline{ME}west \wedge read[f(west)=1]:\textbf{Follow})$\\
 In state \textbf{Wait}:\\
 \Comment{$time1$ is defined as the number of rounds since the agent has encountered an edge along $west$ with $f(west)=0$.}\\
  \textsc{Move}$(west|~time1>1\wedge read[f(west)=0]\wedge \overline{ME}west:\textbf{TerminateW};MEwest:\textbf{Follow})$\\
  In state \textbf{TerminateW}: Terminate, BH is the next node towards $west$.\\
  In state \textbf{Exit}:\\
 Each agent moves until reaches the node in $j$-th column, if a node is marked $\bot$ or $0$ , mark it 1 and continue its movement.\\
 \If{Encounters a Missing Edge}
 {
 If no agent is already stuck, then lowest Id agent gets stuck for the edge to reappear, others change direction.\\
 }
 \end{algorithm2e}

\begin{lemm}\label{lemma:complexityAlg1a1b}
    If $l$ ($l\geq 2$) agents execute \textsc{Cautious-Move}$(west,j)$ along a safe ring $R_i$, then at least $l-1$ agents reach $v_{i,j}$ within $3lm$ rounds.
\end{lemm}

\begin{proof}
   We claim the above statement is true with the help of induction.

   \noindent\textit{Base Case}: Let two agents $a_1$ and $a_2$ perform \textsc{Cautious-Move}$(west,j)$ starting from the node $v_{i,j}$ along $R_i$. So, initially $a_1$ becomes the $Leader$ and $a_2$ becomes the $Follower$ and they start moving \textit{cautiously} along $west$. Now, suppose $a_1$ encounters a missing edge, then in this scenario, according to the algorithm both the agent waits until $time=3m$, after which they try to return to $v_{i,j}$. Now, while the edge along $west$ is still missing, they together start moving towards $east$ to reach back to $v_{i,j}$, and this takes at most $m$ rounds. In this situation, the adversary reappears the earlier edge, and disappears one edge along $east$ further blocking both $a_1$ and $a_2$. Now according to the algorithm, $a_1$ waits whereas $a_2$ changes its direction to $west$ and tries to reach $v_{i,j}$ in at most $m$ rounds, in the meantime, suppose the adversary again stops $a_2$ by disappearing an edge along $west$ whereas reappears the earlier edge unblocking $a_1$, so in additional $m$ rounds, either $a_1$ or $a_2$ reaches $v_{i,j}$. Hence in at most $6m$ rounds from the start of the algorithm, at least $a_1$ or $a_2$, reaches the desired node $v_{i,j}$.

   \noindent\textit{Induction Hypothesis}: Let us suppose $l-2$ agents reach $v_{i,j}$ within $3(l-1)m$ rounds.

   \noindent\textit{Inductive Case}: So, in this situation, the only possibility is that the remaining two agents $a_{l-1}$ and $a_{l}$ (say) are together and encounters a missing edge along $east$ (or $west$) while the last of the other $l-2$ agent, reach $v_{i,j}$ within $3(l-1)m$ rounds along $west$ (or $east$). It is because, in state \textit{Exit}, the adversary can make only the lowest Id agents wait with respect to a missing edge, whereas the other agents start moving in the opposite direction. This means that if, these two remaining agents is not together then either of these two agents will have eventually reached $v_{i,j}$ within $3(l-1)m$ rounds. 
   
   So, suppose these two agents while moving towards $west$ (or $east$) for at most $m$ rounds, encounters a missing edge, which makes the agent $a_{l-1}$ wait, whereas $a_l$ changes its direction to $east$ (or $west$) and tries to reach $v_{i,j}$ in at most $m$ rounds, and again encounters another missing edge along $east$ (or $west$), so in another $m$ rounds either of these two agents reach $v_{i,j}$. This implies in at most $3(l-1)m+3m$ rounds $l-1$ agents reach the desired node. Hence, the inductive step also holds.\end{proof}

\begin{blab}\label{corollary:complexityExitAlg1a1b}
If $l$ agents enter the state \textit{Exit}, then at least $l-1$ agents reach $v_{i,j}$ by at most $3(l-1)m$ rounds.
\end{blab}

\noindent\underline{\textsc{CautiousDoubleOscillation}\cite{BHSDynRingLuna}} We have used this BHS ring exploration algorithm as a sub-routine in our BHS Torus exploration algorithm. The only difference is that both \textsc{Avanguard} and \textsc{Retroguard} while exploring a new node marks the corresponding ports safe or unsafe in whiteboard, so this means if \textsc{Retroguard} enters the black hole while exploring a sector of $\sqrt{m}$ nodes along $R_i$ (or $\sqrt{n}$ nodes along $C_j$) then using the whiteboard instead of a pebble, \textsc{Leader} can detect its location. 

As stated in \cite{BHSDynRingLuna}, three agents executing \textsc{CautiousDoubleOscillation} requires $O(m^{1.5})$ rounds to detect the black hole along a 1-interval connected ring of size $m$.
     
\section{Co-located Agents}\label{section-colocated}
In this section, we propose two BHS algorithms on $n \times m$ dynamic torus. First algorithm requires $n+3$ agents and works in $O(nm^{1.5})$ rounds, whereas the second algorithm requires $n+4$ agents and works in $O(nm)$ rounds.

\subsection{BHS with $n+3$ agents}
The set of $n+3$ agents, $A=\{a_1,a_2,\ldots,a_{n+3}\}$ are initially located at a safe node $v_{i,j}$, also termed as $home$. Initially from $home$, $a_1$ and $a_2$ executes the algorithm \textsc{Cautious-Move}$(north,i)$, whereas $a_3$ and $a_4$ executes \textsc{Cautious-Move}$(south,i)$. Once $12n$ rounds have passed, at least $3$ out of these 4  agents return to $v_{i,j}$ (refer corollary \ref{corollary:complexityExitAlg1a1b}). Whenever 3 among 4 agents return back to $home$, the first three lowest Id agents become \textsc{Leader}, \textsc{Avanguard} and \textsc{Retroguard} and and they are instructed to perform \textsc{CautiousDoubleOscillation} along $R_{i}$. Now, as per Lemmas 14 and 15 in \cite{BHSDynRingLuna} it takes at most $T=19m^{1.5}+7(m+\sqrt{m})$ rounds to locate the black hole along $R_i$. So, after $T$ rounds since the start of \textsc{CautiousDoubleOscillation}, if the algorithm hasn't terminated (or the black hole is not detected) then these agents are instructed to return to $v_{i,j}$ which is the desired node, irrespective of the fact, whether they are stuck or not. While returning, if an agent encounters a missing edge, then the lowest Id agent waits, whereas the other agent changes direction. Using this strategy, in at most $6m$ rounds, at least 2 among 3 agents return to $v_{i,j}$ (as this is similar to state \textit{Exit} in algorithm \textsc{Cautious-Move}(), hence by corollary \ref{corollary:complexityExitAlg1a1b} this bound holds). After which they all together start executing \textsc{Cautious-WaitMoveSouth}$(i-1,4)$, which enables at least 3 among $n+3$ agents reach the node $v_{i-1,j}$ and continue the same process. This process iterates for each $R_i$, where $0\le i \le n-1$. The pseudocode is explained in Algorithm \ref{alg-4}.

\begin{algorithm2e}[!ht]\footnotesize
\caption{BHS with $n+3$ agents}\label{alg-4}
Instruct $a_1$ and $a_2$ to perform \textsc{Cautious-Move}($north,i$).\\
Instruct $a_3$ and $a_4$ to perform \textsc{Cautious-Move}($south,i$).\\
\If{$time>12n$}
{
\For{$t=i;~t\leq i+1; t--$}
{
$3$ lowest Id agents at $v_{t,j}$ perform \textsc{CautiousDoubleOscillation}.\\
\Comment{$time1$ is defined as the number of rounds since the last call of \textsc{CautiousDoubleOscillation}}.\\
\If{$time1>T$}
{
Exit \textsc{CautiousDoubleOscillation}.\\
Return to $v_{i,j}$, whether stuck or not.\\
\If{Encounters a missing edge}
{
Lowest Id agent waits for the edge to reappear, whereas the remaining agents change direction.\\
}
}
\If{$time1>T+6m$}
{
Each agent along $C_j$ is instructed to perform \textsc{Cautious-WaitMoveSouth}$(t-1,4)$.\\
}
}
}
\end{algorithm2e}

\begin{lemm}
    Algorithm \ref{alg-4}, ensures that there always exist 3 agents to perform \textsc{CautiousDoubleOscillation} along $R_i$, where $0\le i \le n-1$.
\end{lemm}
\begin{proof}
   Suppose the black hole is somewhere along the ring $R_{i+1}$, and since $home=v_{i,j}$, so according to Algorithm \ref{alg-4}, $R_{i+1}$ is the last ring to be explored for black hole. Moreover, since $\mathcal{G}$ has $n$ many rows and there are $n+3$ agents. We prove the above statement by contradiction, suppose two agents reach $v_{i+1,j}$ for exploration of $R_{i+1}$ instead of three at the end of \textsc{Cautious-WaitMoveSouth}$(i+1,4)$ at the $(n-1)$-th iteration of Algorithm \ref{alg-4}. This implies that leaving these two agents which has reached $v_{i+1,j}$, there are $n+1$ agents along $n-1$ rows, i.e., either there exists one row with three agents or there exists two rows with two agents each. In either case, after the execution of \textsc{CautiousDoubleOscillation}, the adversary can restrict at most one among 3 agents from returning to a node along $C_j$. So, in each execution along $R_{t}$, $\forall t\in \{0,\ldots,n-1\}\backslash \{i+1\}$ at most one agent is unable to return to $C_j$. This means, $n-1$ agents are unable to return to $C_j$, so the remaining two agents are somewhere along $C_j$. Again, by Theorem \ref{theorem:corrAlg1} at most one agent can be stuck while executing \textsc{Cautious-WaitMoveSouth}$(i+1,4)$, that means at least one among these two agents must have reached $v_{i+1,j}$ as each agent along $C_j$ performs this execution while in the $(n-1)$-th iteration, which leads to a contradiction to the fact that 2 agents reach $R_{i+1}$ at the end of $(n-1)$-th iteration.\end{proof}

   \begin{lemm}\label{lemma:R_iExploreAlg4}
       It takes at most $T+6m+15n$ rounds to perform one iteration of the for loop in Algorithm \ref{alg-4}.
   \end{lemm}

\begin{proof}
    Observe that \textsc{CautiousDoubleOscillation} takes at most $T= 19m^{1.5}+7(m+\sqrt{m})$ rounds on $R_t$ (for any $t\in \{0,1,\ldots,n-1\}$) to detect the black hole (by Lemmas 14 and 15 of algorithm \cite{BHSDynRingLuna}), and if there is no black hole along $R_t$ then within this many rounds the ring is explored. So, whenever $T$ rounds has elapsed since the start of \textsc{CautiousDoubleOscillation} on $R_t$, Algorithm \ref{alg-4} instructs these 3 agents to return to $v_{t,j}$ irrespective of the fact that they are stuck or not, so in at most $6m$ rounds at least 2 agents return to $v_{t,j}$. Lastly, all the agents not stuck along $C_j$ are instructed to perform \textsc{Cautious-MoveSouth}$(t-1,4)$, and by Theorem \ref{theorem:corrAlg1}, it takes at most $15n$ ($g(4)+3n=12n+3n$) rounds for at least 3 agents to reach $v_{t-1,j}$. Hence, in at most $T+6m+15n$ rounds an iteration of for loop of algorithm \ref{alg-4} is executed.\end{proof}
 \begin{lemm}\label{lemma:CorrAlg4}
     Our algorithm correctly locates the black hole.
 \end{lemm}
 \begin{proof}
     Suppose, the black hole is along $C_j$ where $home=v_{i,j}$, then in step 1 and step 2 of Algorithm \ref{alg-4}, the black hole is detected. It is because, $a_1$ and $a_3$ being the $Leader$ along $north$ and $south$, respectively, can simultaneously enter the black hole in the worst case, whereas $a_2$ and $a_4$, being the $Follower$ of them waits for at most $3n$ rounds in their $Leader$'s adjacent nodes. Since $C_j$ is a column ring of size $n$, hence there exists a round, in which $a_2$ is waiting for $a_1$, whereas $a_4$ is waiting for $a_3$, and the adversary can only disappear only one edge among them. So, either of these two agents will inevitably locate the black hole position and terminate the algorithm.

Otherwise, if the black hole is located along a ring $R_t$ such that the node is not along $C_j$, so when the algorithm instructs the 3 agents at $v_{t,j}$ to perform \textsc{CautiousDoubleOscillation} at some iteration, then these three agents ultimately detects the black hole, and it follows from the correctnes of \textsc{CautiousDoubleOscillation} in \cite{BHSDynRingLuna}.\end{proof}

\begin{theorem}\label{theorem:ComplexityAlg4}
    A group of $n+3$ agents along a dynamic torus of size $n\times m$ correctly locates the black hole in $O(nm^{1.5})$ rounds while executing Algorithm \ref{alg-4}.
\end{theorem}

\begin{proof}
    By lemma \ref{lemma:R_iExploreAlg4}, it takes at most $T+6m+15n\approx O(m^{1.5})$ rounds (since $3\leq n \leq m$) to explore a ring $R_i$ and reach the next node $v_{i-1,j}$. Now, since the torus is of size $n\times m$, hence there are $n$ such rings in $\mathcal{G}$ and in order to explore and locate the black hole, in the worst case, each of the $n$ rings need to be explored, so it takes $n\cdot O(m^{1.5})=O(nm^{1.5})$ rounds.\end{proof}

\subsection{BHS with $n+4$ agents}
 In this case the set of $n+4$ agents, $A=\{a_1,a_2,\ldots,a_{n+4}\}$ agents are initially co-located at $home=v_{i,j}$, say. The algorithm in this case is similar as the earlier BHS algorithm with $n+3$ agents, the only difference is that here in order to explore $R_t$, instead of 3, 4 agents are used, where the lowest and second lowest Id agents at $v_{t,j}$ perform \textsc{Cautious-Move}$(west,j)$ and the third lowest and fourth lowest Id agents are instructed to perform \textsc{Cautious-Move}$(east,j)$, instead of \textsc{CautiousDoubleOscillation}. The pseudocode is explained in algorithm \ref{alg-5}.

 \begin{algorithm2e}[!ht]\footnotesize
 \caption{BHS with $n+4$ agents}\label{alg-5}
Instruct $a_1$ and $a_2$ to perform \textsc{Cautious-Move}($north,i$).\\
Instruct $a_3$ and $a_4$ to perform \textsc{Cautious-Move}($south,i$).\\

\If{$time>12n$}
{
 \For{$t=i;~t\leq i-1;~t--$}
 {
 Instruct the lowest and second lowest Id agents at $v_{t,j}$ to perform \textsc{Cautious-Move}($west,j$).\\
 Instruct the third lowest and fourth lowest Id agents at $v_{t,j}$ to perform \textsc{Cautious-Move}($east,j$).\\
\Comment{$time1$ is defined as the time since the last call of \textsc{Cautious-Move}}.\\
 \If{$time1>12m$}
 {Perform \textsc{Cautious-WaitMoveSouth}$(t-1,5)$.}
 }
 }
 \end{algorithm2e}

\begin{lemm}\label{lemma:3among4Alg5}
    At least 3 among 4 agents executing \textsc{Cautious-Move}$(west,j)$ and \textsc{Cautious-Move}$(east,j)$ along $R_t$ at some $i$-th iteration of Algorithm \ref{alg-5} reach $v_{t,j}$ within $12m$ rounds since the start of \textsc{Cautious-Move}() in the current iteration, where $0\leq t \leq n-1$.
\end{lemm}

\begin{proof}
    As stated, \textsc{Cautious-Move}$(west,j)$ and \textsc{Cautious-Move}$(east,j)$, takes at most $3m$ rounds to explore the ring $R_t$, after which each of the 4 agents enter the state \textit{Exit}, so by Corollary \ref{corollary:complexityExitAlg1a1b} at least 3 among 4 agents reach $v_{t,j}$ in at most $9m$ rounds ($9m$ rounds from the moment they entered the state \textit{Exit}). So, in total at least 3 among 4 return to $v_{t,j}$ within $12m$ rounds since the start of execution of \textsc{Cautious-Move}() while exploring $R_t$.
\end{proof}

 \begin{lemm}
     Our BHS algorithm with $n+4$ agents, ensures that in every iteration there always exists 4 agents to perform \textsc{Cautious-Move}$(west,j)$ and \textsc{Cautious-Move}$(east,j)$.
 \end{lemm}  

\begin{proof}
      As earlier stated in Lemma \ref{lemma:3among4Alg5} at least 3 among 4 agents return to $v_{t,j}$ within $12m$ rounds, while exploring $R_t$. So, if $R_{i+1}$ is the last ring to be explored in $\mathcal{G}$, then at most $n-1$ agents are stuck along $R_t$, $t\in \{0,1,\ldots,n-1\}\backslash\{i+1\}$, whereas the remaining 5 agents must have successfully reached the nodes along $C_j$ before the $(n-1)$-th execution of \textsc{Cautious-WaitMoveSouth}$(i,5)$. So, just the execution of \textsc{Cautious-WaitMoveSouth}$(i,5)$ at line 9 of algorithm \ref{alg-5}, ensures that at least 4 among 5 agents reach $v_{i+1,j}$ within $19n$ (by Theorem \ref{theorem:corrAlg1}) rounds from the start of this execution. This in turn proves that at each iteration, there exists at least 4 agents at $v_{i,j}$ in order to explore $R_i$.
\end{proof}

\begin{lemm}
    A set of $n+4$ agents executing Algorithm \ref{alg-5}, correctly locates the black hole.
\end{lemm}

\begin{proof}
     We have the following cases: first, the black hole lies along $C_j$, second, it does not lie along $C_j$.
     \begin{itemize}
         \item If the black hole is along $C_j$, then 4 agents $a_1$, $a_2$, $a_3$ and $a_4$ correctly locates the black hole. It is because, $a_1$ and $a_3$ acts a $Leader$ while exploring along $north$ and $south$, respectively, whereas $a_2$ and $a_4$ acts as a $Follower$ of $a_1$ and $a_3$. Since, $C_j$ has $n$ many nodes, so both $a_1$ and $a_3$ if not blocked due to a missing edge, eventually reaches the black hole node from either side in the worst case, whereas their respective followers reach the adjacent node. While the adversary has the ability to restrict one of the $Follower$ from determining the black hole node, the other $Follower$ inevitably locates the black hole node and terminates the algorithm.
         \item If the black hole node is not along $C_j$, suppose it is at the node $v_{i',j'}$ along $R_{i'}$ (where $i'\in \{0,1,\ldots,n-1\}$ and $j'\in\{0,1,\ldots,m-1\}\backslash\{j\}$), in which case at some $t$-th iteration, i.e., while exploring $R_{i'}$, any set of 4 agents executing \textsc{Cautious-Move}$(west,j)$ and \textsc{Cautious-Move}$(east,j)$, can correctly locate the black hole by the earlier argument.
     \end{itemize}
\end{proof}

\begin{theorem}\label{theorem:n+4colocatedRoundComplexity}
    A group of $n+4$ agents executing Algorithm \ref{alg-5} along a dynamic torus of size $n\times m$ correctly locates the black hole in $O(nm)$ rounds.
\end{theorem}

\begin{proof}
    Initially, to explore $C_j$, 4 agents performing \textsc{Cautious-Move}$(north,i)$ and \textsc{Cautious-Move}$(south,i)$ requires at most $12n$ rounds, either to detect the black hole or to explore $C_j$ and return to $v_{i,j}$. Next, for exploring a ring $R_i$ of size $m$, a set of 4 agents again executing \textsc{Cautious-Move}$(west,j)$ and \textsc{Cautious-Move}$(east,j)$, requires at most $12m$ rounds either to detect the black hole or to explore $R_i$ and then return back to the node $v_{i,j}$. Lastly, in order to move at least 4 agents from $R_i$ to $R_{i-1}$, at most $19n$ ($g(5)+3n=16n+3n$) rounds are required due to the execution of \textsc{Cautious-WaitMoveSouth}$(i-1,5)$ at line 9 of Algorithm \ref{alg-5}. So, in order to successfully explore a single ring for black hole, our algorithm \ref{alg-5} requires at most $12m+19n$ rounds, leaving the initial exploration of $C_j$. Now, there are $n$ such rings to be explored in the worst case, so the total number of rounds required to execute Algorithm \ref{alg-5} is: $12n+n(12m+19n)=O(mn)$, since $3\leq n \leq m$.
\end{proof}

\begin{figure}
  	\centering
  	\includegraphics[scale=0.8]{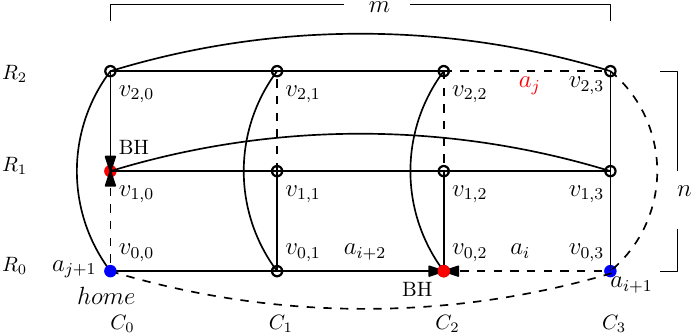}
  	\caption{Representing all possible black hole consumption while executing either algorithm \ref{alg-4} or \ref{alg-5}}
  	\label{fig:black_hole_consumption_part1}
  \end{figure}

\begin{blub}
    Our both Algorithms \ref{alg-4} and \ref{alg-5} for co-located agents, working with a set of $n+3$ and $n+4$ agents, respectively, ensures that at most 2 agents enter the black hole. Fig. \ref{fig:black_hole_consumption_part1} denotes that if the black hole is along the column of $home$ node (i.e., along $C_0$ in the figure) then in the first two steps of both the algorithms the black hole is detected and at most one agent each enters the black hole from $north$ and $south$. Otherwise, if the black hole is not along the column of $home$ (i.e., along $C_2$ (say) in figure), then while exploring the respective row ($R_0$ in the figure) the black hole is detected and at most two agents can enter the black hole along $east$ and $west$. In each possible case, our algorithms ensure at most two agents enter the black hole in the worst case.
\end{blub}

\section{Scattered Agents}\label{section-scattered}

In this section we propose two BHS algorithms on an $n\times m$ dynamic torus. Our first algorithm works with $n+6$ agents and require $O(nm^{1.5})$ rounds, whereas our second algorithm works with $n+7$ agents and require $O(nm)$ rounds.

\subsection{BHS with $n+6$ agents} 

A set of $n+6$ agents, $A=\{a_1,a_2,\ldots,a_{n+6}\}$ are initially scattered along different nodes of the torus $\mathcal{G}$, i.e., the agents are arbitrarily placed, where there may be more than one agent at a node or there can be single agent at each $n+6$ nodes in $\mathcal{G}$. 

At the first step, each agent performs \textsc{Cautious-WaitMoveWest}$(0,6)$ from any initial configuration, so after $23m$ ($g(6)+3m=20m+3m$) rounds has elapsed since the start of \textsc{Cautious-WaitMoveWest}$(0,6)$, each agent currently along $C_0$ is further instructed to perform \textsc{Cautious-WaitMoveSouth}($0,6$), so after $23n$ rounds has elapsed since \textsc{Cautious-WaitMoveSouth}($0,6$), if at least 3 agents have reached the node $v_{0,0}$, then 3 lowest Id agents at $v_{0,0}$ become \textsc{Leader}, \textsc{Avanguard} and \textsc{Retroguard}, respectively and are then instructed to perform \textsc{CautiousDoubleOscillation} along $R_0$. Hence, within $T$ rounds from the start of \textsc{CautiousDoubleOscillation} either the black hole is detected and the algorithm terminates or the ring $R_0$ is explored. After $T$ rounds since the start of \textsc{CautiousDoubleOscillation} these 3 agents are instructed to return to $v_{0,0}$ by marking each node along their movement till $v_{0,0}$ to 1 (as the ring is explored and there is no black hole in this ring, so an agent can mark each port as safe, if not already marked so). So, by corollary \ref{corollary:complexityExitAlg1a1b} in at most $6m$ rounds at least 2 among these 3 agents return to $v_{0,0}$, after which, each agent in $G$ are instructed to perform \textsc{Cautious-WaitMoveWest}$(0,6)$. 

On the other hand, if two agents have reached $v_{0,0}$ after $23n$ rounds has elapsed since \textsc{Cautious-WaitMoveSouth}($0,6$), then the lowest Id agent \textit{cautiously} walks along $west$ whereas the other agent \textit{cautiously} walks along $east$. If along their movement they $catches$ another agent trying to move along the same direction, then they together perform \textsc{Cautious-Move}() in the same direction. After $3m$ rounds has passed since they have started \textit{cautious} walk, these agents along $R_0$ are instructed to return to $v_{0,0}$ by marking each port along their movement to 1. So, within $6m$ rounds, if 3 agents are along $R_0$ then at least 2 among them returns or if 2 agents are along $R_0$ then at least 1 among them returns, further each agent along $G$ is again instructed to perform \textsc{Cautious-WaitMoveWest}$(0,6)$. This process iterates for each $R_i$ rings (where $0\le i \le n-1$), depending on whether 2 or 3 agents have reached the node $v_{i,0}$. The pseudocode is explained in the algorithm \ref{alg-6}.  

 \begin{algorithm2e}[!ht]\footnotesize
 \caption{BHS with $n+6$ agents}\label{alg-6}
 \textsc{Cautious-WaitMoveWest}$(0,6)$.\\
 \For{$i=0;~i\leq n-1;~i++$}
 {
 \Comment{$time1$ is defined as the number rounds since the last call of \textsc{Cautious-WaitMoveWest}().}\\
 \If{$time1>23m$}
 {
 All the agents along $C_0$ perform \textsc{Cautious-WaitMoveSouth}($i,6$).\\ \Comment{$time2$ is defined as the number rounds since the last call of \textsc{Cautious-WaitMoveSouth}().}\\
 \If{$time2>23n$}
 {
 \If{2 agents at $v_{i,0}$}
 {
 Instruct lowest Id agent at $v_{i,0}$ to perform \textit{Cautious} walk along $west$ and the other agent to perform \textit{Cautious} walk along $east$.\\
 If either agent \textit{catches} another agent, then they together perform \textsc{Cautious-Move} along the same direction.\\
 \Comment{$time3$ is defined as the number rounds since the agents started \textit{cautious} walk along $R_i$}\\
\If{$time3>3m$}
{
Each agent along $R_i$ is instructed to return to $v_{i,0}$ while marking each port along their movement to 1 if not already marked, irrespective of the fact that they are stuck or not.\\
\If{encounters a missing edge}
{
Lowest Id agent waits for the edge to reappear, whereas the other agents change direction.\\
}
\If{$time3>9m$}
{Each agent perform \textsc{Cautious-WaitMoveWest}$(0,6)$.\\}
}
}
\Else
{
Instruct the three lowest Id agents at $v_{i,0}$ to perform \textsc{CautiousDoubleOscillation}.\\
\Comment{$time4$ is defined as the number of rounds since the last call of \textsc{CautiousDoubleOscillation}}.\\
\If{$time4>T$}
{
Each agent along $R_i$ is instructed to return to $v_{i,0}$ while marking each port along their movement to 1 if not already marked, irrespective of the fact that they are stuck or not.\\
\If{encounters a missing edge}
{
Lowest Id agent waits for the edge to reappear, whereas the other agents change direction.\\
}
\If{$time4>T+6m$}
{Each agent perform \textsc{Cautious-WaitMoveWest}$(0,6)$.\\}
}
}
}
}
}
\end{algorithm2e}


\begin{lemm}\label{lemma:3agentBHAlg6}
    If 2 agents reach $v_{i,0}$ at the $i$-th iteration of Algorithm \ref{alg-6} when $time2>23n$, and the algorithm has not terminated yet, then this implies exactly 3 agents has entered black hole from three different directions.
\end{lemm}
\begin{proof}
    We prove the above lemma by contradiction, let us consider less than 3 agents have entered the black hole, this implies there are at least $n+4$ remaining agents, out of which only 2 agents has reached $v_{i,0}$ within $time2\leq 23n$ at the $i$-th iteration of Algorithm \ref{alg-6}. Since, $v_{i,0}$ is the node to be visited at the $i$-th iteration, as $R_i$ is the ring to be explored at the current iteration, this means the black hole is not located along $R_t$, $\forall~t\in\{0,1,\ldots,i-1\}$), as otherwise the algorithm must have already terminated while exploring $R_t$. This implies, black hole is somewhere along the nodes of the following row rings $R_i,\ldots, R_{n-1}$. Now, we have two cases:
    \begin{itemize}
        \item Let the black hole be at a node $v_{t',j}$ (where $t'\in \{i,\ldots,n-1\}$ and $j\in \{1, \ldots, m-1\}$). In this case, as $R_{t'}$ is not yet explored, hence we claim that at most one agent has entered the black hole while at most one agent is stuck while trying to reach $v_{t',0}$. Note that this claim holds as the only execution performed along $R_{t'}$ till the current iteration is \textsc{Cautious-WaitMoveWest}$(0,6)$ which assures by Corollary \ref{corollary:2agentsBHAlg1} that if $v_{t',0}$ is not the black hole node then at most 1 agent can enter the black hole. 
         Now, if all row rings have at most 5 agents at the end of $(i-2)$-th iteration, then by Theorem \ref{theorem:corrAlg1}, leaving at most one agent at each row ring and one agent consumed by the black hole, remaining at least 5 ($=n+6-n-1$) agents must reach $C_0$ after execution of \textsc{Cautious-WaitMoveWest}$(0,6)$ at the $i-1$-th iteration. Further, an execution of \textsc{Cautious-WaitMoveSouth}$(0,6)$ in the $i$-th iteration will bring at least 4 among these 5 agents to $v_{i,0}$ within $time2\leq 18n$, contradicting our claim. On the other hand, if there exists at least one row ring with at least 6 agents at the end of $(i-2)$-th iteration and other $n-1$ rows with $n$ agents, then again by Theorem \ref{theorem:corrAlg1}, at least 5 among them must reach $C_0$ (if that ring with at least 6 agents is $R_{t'}$ then leaving at most 2 agents, i.e., one in black hole and one stuck, the remaining 4 agents must reach $C_0$ while executing \textsc{Cautious-WaitMoveWest}$(0,6)$, whereas at least one agent among remaining $n$ agents distributed in the other $n-1$ rings must also reach $C_0$ while performing the same \textsc{Cautious-WaitMoveWest}$(0,6)$ execution in the $(i-1)$-th iteration) by the end of $(i-1)$-th iteration and further execution of \textsc{Cautious-WaitMoveSouth}$(0,6)$ in the $i$-th iteration takes at least 4 among them to $v_{i,0}$ within $time2\leq 23n$, contradicts our claim.

        \item Let the black hole be at a node $v_{t',0}$ (where $t'\in \{i,\ldots,n-1\}$). In this case, if 4 agents operate along $R_{t'}$ then execution of \textsc{Cautious-WaitMoveWest}$(0,6)$ will itself detect the black hole (refer Corollary \ref{corollary:Possibile4agentAlg1}). So, this means there are at most 3 agents along $R_{t'}$. If 3 agents are along $R_{t'}$ then by Lemma \ref{lemma:Impossibile3agentAlg1}, 2 of them can enter the black hole from $west$ and $east$ while the other agent can remain stuck somewhere along $R_{t'}$. Otherwise, if at most 2 agents are along $R_{t'}$ then also either both of these agents have entered black hole along $east$ and $west$ or one among them has entered the black hole and the other is stuck, while another agent must have entered the black hole along $north$ or $south$ while executing \textsc{Cautious-WaitMoveSouth}$(0,6)$ at some iteration of Algorithm \ref{alg-6}. In either case, there exists at least $n+3$ ($=n+6-2-1$) agents along the remaining $n-1$. By earlier argument, if there exists at least one row with 6 agents (and that row is not $R_{t'}$) at $(i-1)$-th iteration then within $time2\leq 23n$ at least 4 among these 6 agents reach the desired node, which leads to a contradiction. So, each of these $n-1$ rows, has at most 5 agents, this means leaving at most $n-1$ agents stuck along $n-1$ row rings (leaving $R_{t'}$), the remaining 4 ($=n+3-(n-1)$) agents must reach $C_0$ at the end of $(i-1)$-th iteration and again execution of \textsc{Cautious-WaitMoveSouth}$(0,6)$ along $C_0$ at the $i$-th iteration must bring at least 3 agents at $v_{i,0}$ within $time2\leq 23n$, contradicts our claim.
    \end{itemize}
    So, in each case we attain a contradiction.\end{proof}

\begin{blab}\label{corollary:3directionsBHAlg6}
    If 3 agents enter black hole from 3 directions without detecting it, then 2 among these 3 directions are $east$ and $west$, whereas the 3rd is either $north$ or $south$. 
\end{blab}

\begin{lemm}\label{lemma:presenceRiAlg6}
    Our BHS algorithm with $n+6$ agents, ensures that if at the $i$-th iteration after $time2>23n$ only 2 agents are present at $v_{i,0}$, then there exists another agent stuck somewhere along $R_i$.
\end{lemm}
\begin{proof}
    This situation of 2 agents reaching $v_{i,0}$ arises when 3 agents have already entered the black hole within $time2\leq 23n$ at the $i$-th iteration of the algorithm, by Lemma \ref{lemma:3agentBHAlg6}. Let us consider that no agent is stuck along $R_i$, so leaving the two agents which has reached $v_{i,0}$, there exists $n+1$ ($n+6-3-2$) agents along $n-1$ row rings. If there exists at least one row (note, that row is not the row with black hole) with at least 6 agents before the beginning of \textsc{Cautious-WaitMoveWest}$(0,6)$ step at the $(i-1)$-th iteration of Algorithm \ref{alg-6}, then this ensures that at least 5 among these 6 agents will reach a node along $C_0$. Next, performing \textsc{Cautious-WaitMoveSouth}$(i,6)$ at the $i$-th iteration will move at least 4 among these 5 agents to $v_{i,0}$ (as at most one can get stuck along $C_0$). This leads to a contradiction, that only 2 agents reach $v_{i,0}$ at the $i$-th iteration. On the contrary, if at most 5 agents are located at each of these $n-1$ rows, leaving $R_i$, before the execution of \textsc{Cautious-WaitMoveWest}$(0,6)$ at the $(i-1)$-th iteration. In that scenario, executing \textsc{Cautious-WaitMoveWest}$(0,6)$, ensures that leaving at most 1 agent in each $n-1$ row ring, the remaining excess agents reach $C_0$. Now, as per our assumption, 3 agents have already entered the black hole, and at most $n-1$ agents are stuck along the $n-1$ row rings, this leaves the excess 4 agents ($=n+6-(n-1)-3$) to be along $C_0$ at the start of $i$-th iteration. So, whenever the algorithm again instructs each agent along $C_0$ to perform \textsc{Cautious-WaitMoveSouth}$(0,6)$ (at step 5 of Algorithm \ref{alg-6}), then at least 3 among them reach $v_{i,0}$ by $time2\leq 23n$ rounds. So, this again leads to a contradiction that 2 agents are at $v_{i,0}$ when $time2>23n$. Hence, this concludes the fact that if 2 agents reach $v_{i,0}$ at the $i$-th iteration , then there exists another agent stuck along $R_i$.\end{proof}

\begin{lemm}
    Our BHS algorithm with $n+6$ agents correctly locates the black hole.
\end{lemm}

\begin{proof}
    We consider each possible position of the black hole, and in each case we prove that our algorithm correctly detects it.
    \begin{itemize}
        \item Suppose black hole is along $C_0$, in this case, consider the black hole node to be at the node $v_{i,0}$. Now, initially if 4 agents are located along $R_i$, then by Corollary \ref{corollary:Possibile4agentAlg1}, 4 agents performing \textsc{Cautious-WaitMoveWest}$(0,6)$ along $R_i$ can locate the black hole in at most $15m$ rounds (refer Corollary \ref{corollary:Possibile4agentAlg1}). Otherwise, if at most 3 agents are located along $R_i$, so the worst possible situation is that 2 agents enter black hole along $east$ and $west$ whereas the 3rd agent is stuck by a missing edge, while executing \textsc{Cautious-WaitMoveWest}$(0,6)$. In this situation, in the $i$-th iteration, while performing \textsc{Cautious-WaitMoveSouth}$(i,6)$ at least 2 agents try to reach $v_{i,0}$ along $C_0$ because at most $n$ agents are stuck along $n$ row rings, at most 3 agents have already entered the black hole without yet detecting it (these 3 directions are $east$, $west$ and $north$ or $south$), and at most 1 agent is stuck along $C_0$ along its way to reach $v_{i,0}$. So, the adversary has no other ability to stop these two agents from reaching $v_{i,j}$, and at most one among them enters the black hole while the other agent correctly locates the black hole.
        \item Suppose the black hole is not along $C_0$, in this case, consider the black hole node to be at a node $v_{i,j}$ (where $j \in \{1,\ldots,m-1\}$). As, the black hole is not along $C_0$, so the only operation performed along $R_i$ before the exploration of $R_i$ is \textsc{Cautious-WaitMoveWest}$(0,6)$. This implies at most one agent is consumed by the black hole either from $east$ or $west$ while performing \textsc{Cautious-WaitMoveWest}$(0,6)$, as by Corollary \ref{corollary:2agentsBHAlg1}, at most 2 agents enter the black hole while performing \textsc{Cautious-WaitMoveWest}$(j,l)$ along $R_i$ when the black hole is itself the node $v_{i,j}$. This means, at the $i$-th iteration at least 3 agents reach the node $v_{i,0}$ within $time2\leq 23n$ and they perform \textsc{CautiousDoubleOscillation} along $R_i$ which correctly locates the black hole.  \end{itemize}\end{proof}

\begin{theorem}\label{theorem:n+6scatteredRoundComplexity}
    A group of $n+6$ agents executing Algorithm \ref{alg-6} along a dynamic torus of size $n\times m$ correctly locates the black hole in $O(nm^{1.5})$ rounds.
\end{theorem}

\begin{proof}
    In order to reach at least 3 agents from any initial configuration to $v_{0,0}$, our algorithm takes at most $23(n+m)$ rounds (since, \textsc{Cautious-WaitMoveWest}$(0,6)$ takes at most $23m$ rounds to reach a node along $C_0$ and \textsc{Cautious-WaitMoveSouth}$(0,6)$ takes another at most $23n$ rounds to reach $v_{0,0}$). Next, performing \textsc{CautiousDoubleOscillation} along $R_0$ takes at most $T$ rounds (where, $T= 19m^{1.5}+7(m+\sqrt{m})$ rounds). Now, there are $n$ many such row rings, and we iterate the above process for each $n$ rings, so total number of rounds is: $n(T+23(n+m))=O(nm^{1.5})$, as $3\leq n \leq m$. \end{proof}

\subsection{BHS with $n+7$ agents}

In this case the set of $n+7$ agents, $A=\{a_1,a_2,\ldots,a_{n+7}\}$ are scattered along the nodes of $\mathcal{G}$. The BHS algorithm with $n+7$ agents is almost similar as the earlier BHS algorithm with $n+6$ agent. The differences are as follows: at each iteration the agents are instructed to perform \textsc{Cautious-WaitMoveWest}$(0,7)$ instead of \textsc{Cautious-WaitMoveWest}$(0,6)$. Next, while exploring a ring $R_i$ at the $i$-th iteration at least 3 agents reach $v_{i,0}$ within $time2> 27m$ ($g(7)+3m=24m+3m$), whereas in Algorithm \ref{alg-6} at least 2 agents reach $v_{i,0}$, within $time2>23m$. Next, if 3 agents reach $v_{i,0}$, then our earlier algorithm, 2 agent scenario is similar to 3 agent scenario in this case. In Algorithm \ref{alg-6}, both agents are instructed to walk \textit{cautiously} along $west$ and $east$, respectively, but now as we have one more agent, so two lowest Id agents among them perform \textsc{Cautious-Move}($west,i$), while the other walks \textit{cautiously} along $east$. Otherwise, if 4 agent reach $v_{i,0}$, then this case is again similar to our 3 agent case in Algorithm \ref{alg-6}, in which these 3 agents perform \textsc{CautiousDoubleOscillation} whereas in Algorithm \ref{alg-7} as we have one more agent, so two lowest Id agents perform \textsc{Cautious-Move}$(west,0)$ and the other two agents (i.e., 3rd lowest and 4th lowest Id agents) perform \textsc{Cautious-Move}$(east,0)$, and all these process iterates for each row ring.

\begin{algorithm2e}[!ht]\footnotesize
\caption{BHS with $n+7$ agents}\label{alg-7}
\textsc{Cautious-WaitMoveWest}$(0,7)$\\
\For{$t=0;~i\leq n-1;i++$}
{
\Comment{$time1$ is defined as the number of rounds since the last call of \textsc{Cautious-WaitMoveWest}()}.\\
\If{$time1>27m$}
{
All the agents along $C_0$ perform \textsc{Cautious-WaitMoveSouth}$(i,7)$.\\
\Comment{$time2$ is defined as the number of rounds since the last call of \textsc{Cautious-WaitMoveSouth}()}.\\
\If{$time2>27n$}
{
\If{3 agents at $v_{i,0}$}
{
Instruct the first two lowest Id agents to perform \textsc{Cautious-Move}$(west,0)$ and the other agent to perform \textit{cautious} walk along $east$.\\
If the single agent performing \textit{cautious} walk along $east$ \textit{catches} another agent not stuck, then they together perform \textsc{Cautious-Move}$(east,0)$.\\
\Comment{$time3$ is defined as the number of rounds since the last call of \textsc{Cautious-Move}$(west,0)$}\\
\If{$time3>3m$}
{
\If{the single agent still performs \textit{cautious} walk along $east$ and not \textsc{Cautious-Move}$(east,0)$}
{
Instruct it to return back to $v_{i,0}$ whether stuck or not by marking each port along its movement to 1, if not already marked.\\
}
\If{$time3>12m$}
{
Each agent is instructed to perform \textsc{Cautious-WaitMoveWest}$(0,7)$.\\
}
}
}
\Else
{
Instruct the two lowest Id agents to perform \textsc{Cautious-Move}$(west,0)$ and the third and fourth lowest Id agent to perform \textsc{Cautious-Move}$(east,0)$.\\
\Comment{$time4$ is defined as the number of rounds since the last call of \textsc{Cautious-Move}$(west,0)$.}\\
\If{$time4>12m$}
{
Perform \textsc{Cautious-WaitMoveWest}$(0,7)$\\
}
}
}
}
}
\end{algorithm2e}

\begin{lemm}\label{lemma:3agentBHAlg7}
    If 3 agents reach $v_{i,0}$ when $time2>27n$, and the algorithm has not terminated, then 3 agents have entered the black hole from three different directions.
\end{lemm}

\begin{lemm}\label{lemma:presenceRiAlg7}
    Our BHS algorithm with $n+7$ agents ensures that if at the $i$-th iteration after $time2>27n$ only 3 agents are present at $v_{i,0}$, then there exists another agent stuck somewhere along $R_i$.
\end{lemm}

Lemmas \ref{lemma:3agentBHAlg7} and \ref{lemma:presenceRiAlg7} are just a consequence of Lemmas \ref{lemma:3agentBHAlg6} and \ref{lemma:presenceRiAlg6}, it is because, the strategies used in the earlier algorithm from reaching $C_0$ then to $v_{i,0}$ at the $i$-th iteration is same in this algorithm as well, only difference is that we have used one extra agent, which changes only the row ring exploration strategies. Also, Corollary \ref{corollary:3directionsBHAlg6} holds for this algorithm as well.

\begin{lemm}
    Algorithm \ref{alg-7} correctly locates the black hole with a set of $n+7$ agents.
\end{lemm}

\begin{proof}
    Similarly as proved earlier in case of Algorithm \ref{alg-6}, in this case also we show that for each possible position of the black hole, our Algorithm \ref{alg-7} correctly locates the black hole.
    \begin{itemize}
        \item Suppose the black hole is along $C_0$, and let the black hole be at the node $v_{i,0}$. If 4 agents are initially placed along $R_i$, then just performing \textsc{Cautious-WaitMoveWest}$(0,7)$ will itself locate the black hole (refer Corollary \ref{corollary:Possibile4agentAlg1}). On the contrary, if at most 3 agents are along $R_i$, then in the worst case two agents get consumed along $east$ and $west$ and the third agent is stuck by a missing edge. Further, at the $i$-th iteration, performing \textsc{Cautious-WaitMoveSouth}$(i,7)$ takes at least 3 agents to $v_{i,0}$. By Lemmas \ref{lemma:3agentBHAlg7} and \ref{lemma:presenceRiAlg7}, this situation occurs when 3 agents have already entered the black hole along $east$, $west$ and $north$ or $south$ while $n$ agents are stuck along $n$ row ring and 1 agent is stuck along $C_0$. Now, observe the adversary cannot stop any more agent along $C_0$ as it has already stopped one agent which is stuck along $C_0$, so these three agents will ultimately detect the black hole position.
        \item Suppose the black hole is not along $C_0$, and let the black hole node be at the node $v_{i,j}$ (where $j \in \{1,2,\ldots,m-1\})$. So, the only operation performed along $R_i$ before its exploration at the $i$-th iteration is \textsc{Cautious-WaitMoveWest}$(0,7)$, which can consume at most one agent and stuck another agent (refer Corollary \ref{corollary:2agentsBHAlg1} and Theorem \ref{theorem:corrAlg1}). This implies at the $i$-th iteration at least 4 agents reach $v_{i,0}$ within $time2\leq 27n$, and performing \textsc{Cautious-Move}$(west,0)$ by 2 agents and \textsc{Cautious-Move}$(east,0)$ by the other 2, will ultimately locate the black hole in at most $3m$ rounds.
    \end{itemize}
\end{proof}

\begin{theorem}\label{theorem:n+7scatteredRoundComplexity}
    A group of $n+7$ agents executing Algorithm \ref{alg-7} along a dynamic torus of size $n\times m$ correctly locates the black hole in $O(nm)$ rounds.
\end{theorem}

\begin{proof}
    From any initial configuration, in order to reach $v_{0,0}$ for at least 3 agents, it takes at most $27(m+n)$ rounds (as \textsc{Cautious-WaitMoveWest}$(0,7)$ and \textsc{Cautious-WaitMoveSouth}$(0,7)$ takes at most $27m$ and $27n$ rounds, respectively). Next, whenever at least 3 agents reach $R_0$, they perform \textsc{Cautious-Move}() which takes another at most $12m$ rounds (where $3m$ rounds are required to explore $R_0$ and $9m$ rounds are required for at least 3 among 4 agents to return to $v_{0,0}$ by Corollary \ref{corollary:complexityExitAlg1a1b}). Hence, in total at most $39m+27n$ rounds are required to explore $R_0$ from any initial configuration. Now, this process is iterated for each $n$ row ring, hence the total number of rounds required to execute Algorithm \ref{alg-7} is: $n(39m+27n)=O(mn)$. \end{proof}

\begin{figure}
  	\centering
  	\includegraphics[scale=0.8]{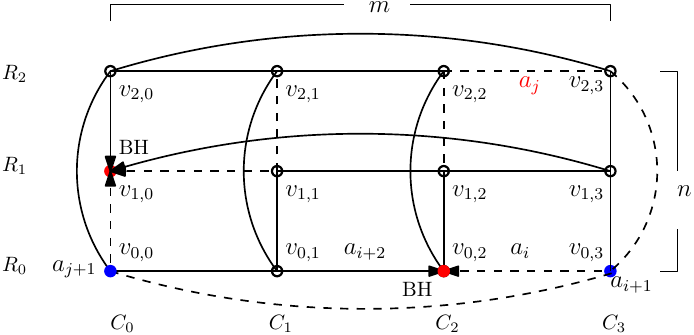}
  	\caption{Represents all possible black hole consumption while executing either Algorithm \ref{alg-6} or \ref{alg-7}}
  	\label{fig:black_hole_consumption_part2}
  \end{figure}

\begin{blub}
    Our both Algorithms \ref{alg-6} and \ref{alg-7} for scattered agents, working with a set of $n+6$ and $n+7$ agents, respectively, ensures that at most 4 agents enter the black hole. Fig. \ref{fig:black_hole_consumption_part2} denotes if the black hole is along $C_0$ then in the first step while executing \textsc{Cautious-WaitMoveWest}() at most two agents can enter black hole from $east$ and $west$, respectively, and if the black hole is not yet detected, then while executing the exploration of $R_i$ (i.e., $R_1$ in figure) another at most two agents can enter the black hole from $north$ and $south$ in step 5 of both algorithms. Otherwise, if the black hole is not along $C_0$, then while exploring the respective row ($R_0$ in the figure) the black hole is detected and at most two agents can enter it along $east$ and $west$, for both the algorithms. So our algorithms ensure at most 4 agents can enter the black hole.
\end{blub}

 \section{Conclusion}\label{section-conclusion}   
In this paper we have considered the black hole search problem on a dynamic torus, in which each row and column are 1-interval connected. We have considered two types of initial configuration of the agents and in each case gave the bounds (both upper and lower bound) on number of agents and complexity in order to locate the black hole. To be specific, when the agents are initially co-located, first, we give lower bounds of $n+2$ and $\Omega(m\log n)$ on number of agents and rounds, respectively. Next, with $n+3$ agents we design a BHS algorithm that works in $O(nm^{1.5})$ rounds whereas with $n+4$ agents we propose an improved algorithm that works in $O(nm)$ rounds. 

When the agents are scattered, we give a lower bound of $n+3$ and $\Omega(mn)$ on number of agents and rounds, respectively. Next, we propose two BHS algorithms, first, works with $n+6$ agent in $O(nm^{1.5})$ rounds and second, works with $n+7$ agents in $O(nm)$ rounds (round optimal algorithm).

In this paper we have considered that each node in the dynamic torus is labeled. A possible future work is to remove this assumption and give a BHS algorithm and check if the bounds get changed. Secondly, for both these cases, finding an agent optimal algorithm is another possible direction which can be pondered in to.
\bibliography{New-Arxiv}

\begin{thebibliography}{10}

\bibitem{TimeOptimalBHSRingBalamohan}
Balasingham Balamohan, Paola Flocchini, Ali Miri, and Nicola Santoro.
\newblock Time optimal algorithms for black hole search in rings.
\newblock {\em Discrete Mathematics, Algorithms and Applications},
  3(04):457--471, 2011.

\bibitem{BHSdynamicCactusAdri}
Adri Bhattacharya, Giuseppe~F Italiano, and Partha~Sarathi Mandal.
\newblock Black hole search in dynamic cactus graph.
\newblock {\em arXiv preprint arXiv:2311.10984}, 2023.

\bibitem{SemiSynchronousExplorationBrandt}
Sebastian Brandt, Jara Uitto, and Roger Wattenhofer.
\newblock A tight lower bound for semi-synchronous collaborative grid
  exploration.
\newblock {\em Distributed Computing}, 33:471--484, 2020.

\bibitem{BHSscatteredStaticTorusChalopin}
J{\'e}r{\'e}mie Chalopin, Shantanu Das, Arnaud Labourel, and Euripides Markou.
\newblock Black hole search with finite automata scattered in a synchronous
  torus.
\newblock In {\em Distributed Computing: 25th International Symposium, DISC
  2011, Rome, Italy, September 20-22, 2011. Proceedings 25}, pages 432--446.
  Springer, 2011.

\bibitem{BHSscatteredSynchronousRingchalopin}
J{\'e}r{\'e}mie Chalopin, Shantanu Das, Arnaud Labourel, and Euripides Markou.
\newblock Tight bounds for black hole search with scattered agents in
  synchronous rings.
\newblock {\em Theoretical Computer Science}, 509:70--85, 2013.

\bibitem{GraphExplorationCohen}
Reuven Cohen, Pierre Fraigniaud, David Ilcinkas, Amos Korman, and David Peleg.
\newblock Label-guided graph exploration by a finite automaton.
\newblock {\em ACM Transactions on Algorithms (TALG)}, 4(4):1--18, 2008.

\bibitem{BHSdirectedGraphczyzowicz}
Jurek Czyzowicz, Stefan Dobrev, Rastislav Kr{\'a}lovi{\v{c}}, Stanislav
  Mikl{\'\i}k, and Dana Pardubsk{\'a}.
\newblock Black hole search in directed graphs.
\newblock In {\em Structural Information and Communication Complexity: 16th
  International Colloquium, SIROCCO 2009, Piran, Slovenia, May 25-27, 2009,
  Revised Selected Papers 16}, pages 182--194. Springer, 2010.

\bibitem{BHSsynchronousTreeczyzowicz}
Jurek Czyzowicz, Dariusz Kowalski, Euripides Markou, and Andrzej Pelc.
\newblock Searching for a black hole in synchronous tree networks.
\newblock {\em Combinatorics, Probability and Computing}, 16(4):595--619, 2007.

\bibitem{TreeExplorationDas}
Shantanu Das, Dariusz Dereniowski, and Christina Karousatou.
\newblock Collaborative exploration of trees by energy-constrained mobile
  robots.
\newblock {\em Theory of Computing Systems}, 62:1223--1240, 2018.

\bibitem{BHSDynRingLuna}
Giuseppe~Antonio Di~Luna, Paola Flocchini, Giuseppe Prencipe, and Nicola
  Santoro.
\newblock Tight bounds for black hole search in dynamic rings.
\newblock {\em arXiv preprint arXiv:2005.07453}, 2020.

\bibitem{ExplorationPebbleDisser}
Yann Disser, Jan Hackfeld, and Max Klimm.
\newblock Tight bounds for undirected graph exploration with pebbles and
  multiple agents.
\newblock {\em Journal of the ACM (JACM)}, 66(6):1--41, 2019.

\bibitem{BHSTokenDobrev}
S~Dobrev, P~Flocchini, R~Kralovic, and N~Santoro.
\newblock Exploring a dangerous unknown graph using tokens.
\newblock In {\em Proceedings of 5th IFIP International Conference on
  Theoretical Computer Science}, pages 131--150, 2006.

\bibitem{BHSCommonInterconnectionNetworkDobrev}
Stefan Dobrev, Paola Flocchini, Rastislav Kr{\'a}lovi{\v{c}},
  P~Ru{\v{z}}i{\v{c}}ka, Giuseppe Prencipe, and Nicola Santoro.
\newblock Black hole search in common interconnection networks.
\newblock {\em Networks: An International Journal}, 47(2):61--71, 2006.

\bibitem{BHSStaticArbNetworkDobrev}
Stefan Dobrev, Paola Flocchini, Giuseppe Prencipe, and Nicola Santoro.
\newblock Searching for a black hole in arbitrary networks: Optimal mobile
  agents protocols.
\newblock {\em Distributed Computing}, 19:1--99999, 2006.

\bibitem{BHSWhiteboardDobrev}
Stefan Dobrev, Paola Flocchini, Giuseppe Prencipe, and Nicola Santoro.
\newblock Mobile search for a black hole in an anonymous ring.
\newblock {\em Algorithmica}, 48:67--90, 2007.

\bibitem{BHSAsyncRingTokensDobrev}
Stefan Dobrev, Rastislav Kr{\'a}lovi{\v{c}}, Nicola Santoro, and Wei Shi.
\newblock Black hole search in asynchronous rings using tokens.
\newblock In {\em Algorithms and Complexity: 6th Italian Conference, CIAC 2006,
  Rome, Italy, May 29-31, 2006. Proceedings 6}, pages 139--150. Springer, 2006.

\bibitem{BHSscatteredUnorientedRingTokenDobrev}
Stefan Dobrev, Nicola Santoro, and Wei Shi.
\newblock Using scattered mobile agents to locate a black hole in an
  un-oriented ring with tokens.
\newblock {\em International Journal of Foundations of Computer Science},
  19(06):1355--1372, 2008.

\bibitem{AsyncExplorationRingFlocchini}
Paola Flocchini, David Ilcinkas, Andrzej Pelc, and Nicola Santoro.
\newblock Computing without communicating: Ring exploration by asynchronous
  oblivious robots.
\newblock {\em Algorithmica}, 65:562--583, 2013.

\bibitem{BHSPebblesFlocchini}
Paola Flocchini, David Ilcinkas, and Nicola Santoro.
\newblock Ping pong in dangerous graphs: Optimal black hole search with
  pebbles.
\newblock {\em Algorithmica}, 62:1006--1033, 2012.

\bibitem{DynamicGeneralGraphExplorationGotoh}
Tsuyoshi Gotoh, Paola Flocchini, Toshimitsu Masuzawa, and Nicola Santoro.
\newblock Exploration of dynamic networks: tight bounds on the number of
  agents.
\newblock {\em Journal of Computer and System Sciences}, 122:1--18, 2021.

\bibitem{DynTorusExpGotoh}
Tsuyoshi Gotoh, Yuichi Sudo, Fukuhito Ooshita, Hirotsugu Kakugawa, and
  Toshimitsu Masuzawa.
\newblock Exploration of dynamic tori by multiple agents.
\newblock {\em Theoretical Computer Science}, 850:202--220, 2021.

\bibitem{RingExplorationNagahama}
Shota Nagahama, Fukuhito Ooshita, and Michiko Inoue.
\newblock Ring exploration of myopic luminous robots with visibility more than
  one.
\newblock {\em Information and Computation}, 292:105036, 2023.

\bibitem{ExplorationFirstWorkShannon}
Claude~E Shannon.
\newblock Presentation of a maze-solving machine.
\newblock {\em Claude Elwood Shannon Collected Papers}, pages 681--687, 1993.

\bibitem{BHSTokenShi}
Wei Shi, Joaquin Garcia-Alfaro, and Jean-Pierre Corriveau.
\newblock Searching for a black hole in interconnected networks using mobile
  agents and tokens.
\newblock {\em Journal of Parallel and Distributed Computing},
  74(1):1945--1958, 2014.

\bibitem{WhiteboardExplorationSudo}
Yuichi Sudo, Daisuke Baba, Junya Nakamura, Fukuhito Ooshita, Hirotsugu
  Kakugawa, and Toshimitsu Masuzawa.
\newblock A single agent exploration in unknown undirected graphs with
  whiteboards.
\newblock {\em IEICE Transactions on Fundamentals of Electronics,
  Communications and Computer Sciences}, 98(10):2117--2128, 2015.

\end{thebibliography}

\end{document}